\documentclass[final,onefignum,onetabnum]{siamart190516}

\usepackage{lipsum}
\usepackage{amsfonts}
\usepackage{amsopn}
\usepackage{braket}
\usepackage{graphicx}
\usepackage{epstopdf}
\usepackage{algorithmic}

\usepackage{physics}
\usepackage{IEEEtrantools}

\usepackage{tikz}
\usetikzlibrary{quantikz}

\DeclareMathOperator{\HH}{\mathcal{H}}
\DeclareMathOperator{\KK}{\mathcal{K}}

\DeclareMathOperator{\M}{\textsc{M}}

\DeclareMathOperator{\D}{\textsc{D}}

\DeclareMathOperator{\PPT}{\textsc{PPT}}
\DeclareMathOperator{\Sep}{\textsc{Sep}}
\DeclareMathOperator{\PSD}{\mathcal{PSD}}
\DeclareMathOperator{\SEP}{\mathcal{SEP}}
\DeclareMathOperator{\x}{\tt x}
\DeclareMathOperator{\y}{\tt y}

\DeclareMathOperator{\CP}{\boldsymbol{CP}}
\DeclareMathOperator{\bP}{\boldsymbol{P}}

\ifpdf
  \DeclareGraphicsExtensions{.eps,.pdf,.png,.jpg}
\else
  \DeclareGraphicsExtensions{.eps}
\fi

\def\RR{{\mathbb R}}
\def\PP{{\mathbb P}}
\def\CC{{\mathbb C}}
\def\NN{{\mathbb N}}
\DeclareMathOperator\Id{Id}

\newcommand{\ols}[1]{\mskip.5\thinmuskip\overline{\mskip-.5\thinmuskip {#1} \mskip-.5\thinmuskip}\mskip.5\thinmuskip}

\newcommand{\invertleadsto}{%
\mathrel{  \reflectbox{\rotatebox[origin=c]{0}{$\leadsto$}}}}

\newsiamremark{remark}{Remark}
\newsiamremark{exmp}{Example}
\newsiamremark{hypothesis}{Hypothesis}
\crefname{hypothesis}{Hypothesis}{Hypotheses}
\newsiamthm{claim}{Claim}

\headers{New families of entangled states}{A. Buckley}

\title{New families of entangled states on 
{\boldmath$\CC^3 \otimes \CC^3$}} 
\author{Anita Buckley\thanks{ Faculty of Mathematics and Physics, University of Ljubljana, Jadranska 19,
SI-1000 Ljubljana, Slovenia 
  (\email{anita.buckley@fmf.uni-lj.si}) and 
  Faculty of Informatics, Universit\`a della Svizzera italiana, 
  Via Buffi 13, 6900 Lugano, Switzerland (\email{anita.buckley@usi.ch}).}
}

\ifpdf
\hypersetup{
  pdftitle={Entangled states},
  pdfauthor={A. Buckley}
}
\fi

\begin{document}

\maketitle

% REQUIRED
\begin{abstract}
We build upon our previous work, the Buckley-\v Sivic method for simultaneous construction of families of positive maps on $3 \times 3$ self-adjoint matrices by prescribing a set of complex zeros to the associated forms. 
Positive maps that are not completely positive can be used to prove (witness) 
that certain mixed states are entangled.
We obtain entanglement witnesses that are indecomposable and belong to extreme rays of the cone of positive maps. 
Consequently our semidefinite program returns 
new examples of entangled states whose entanglement cannot be certified by the transposition map nor by other well-known positive maps.
The constructed states as well as the method of their construction offer some valuable insights for quantum information theory, in particular into the geometry of positive cones.
\end{abstract}

% REQUIRED

% REQUIRED

\section{Introduction}

\subsection{Notation and context} \label{subsec:Notation}
In quantum information theory we typically consider complex finite-dimensional Hilbert spaces and operators between them. We use the following notation: let $A \in B(\HH, \KK)$ be an operator 
(i.e., a linear map, the notion $B$ stands for boundedness)
from $\HH$ to $\KK$, its unique adjoint is denoted by  $A^{\dagger} \in B(\KK, \HH)$, and when $\HH=\KK$ and $A^{\dagger} =A$ we call $A\in B(\HH)$ a \textit{self-adjoint} operator. 
The space of self-adjoint operators $B^{\text{sa}}(\HH)$ is a real (however not complex) vector subspace in $B(\HH)$.
When working with
complex Hilbert spaces it is standard to 
use Dirac's bra-ket notation, where a vector is written as a \textit{ket} vector $\ket{\psi} \in \HH$ and the same vector considered in the dual space 
(i.e., $\ket{\psi}^{\dagger} \in \HH^{\ast}$) is 
written as a \textit{bra} vector $\bra{\psi}$.
A \textit{quantum state} on $\HH$ is a trace one positive self-adjoint operator. In literature quantum states are alternatively referred to as \textit{density operators} or \textit{density matrices}, 
therefore we denote the convex set of states by $\D(\HH)$. By the Krein-Milman theorem, $\D(\HH)$ is the convex hull of its extreme points called \textit{pure states}, i.e., states of the form $\ketbra{\psi}{\psi}$, where $\psi \in \HH$ is a norm one vector.
A state on a multipartite Hilbert space 
\begin{equation*}
    \HH=\HH_1 \otimes \HH_2 \otimes \cdots \otimes \HH_k
\end{equation*}
is by definition \textit{pure separable} if it is of rank one $\ketbra{\psi}{\psi}$ 
with $\psi=\psi_1 \otimes \cdots \otimes \psi_k$ for some unit vectors 
$\psi_i \in \HH_i$; and it is a
 \textit{separable state} if it can be written as a convex combination of pure separable states. A state is \textit{entangled} if it is not separable. In other words,
the set of separable states is
\begin{IEEEeqnarray*}{CCC}
\Sep(\HH) & = & \text{conv}  
\set{ \ketbra{\psi_1 \otimes \cdots \otimes \psi_k}
{\psi_1 \otimes \cdots \otimes \psi_k}   \colon 
\psi_i \in \HH_i} \\
 & = & 
 \text{conv}  
\set{ \rho_1 \otimes \cdots \otimes \rho_k
 \colon 
\rho_i \in \D(\HH_i) };
\end{IEEEeqnarray*}
this follows
from the canonical identifications between
$B\left( \otimes_{i=1}^k \HH_i \right) = \otimes_{i=1}^k B(\HH_i)$
and  between the real vector spaces
\begin{equation*}
    B^{\text{sa}}\left( \bigotimes_{i=1}^k \HH_i \right) = 
    \bigotimes_{i=1}^k B^{\text{sa}}(\HH_i)
\end{equation*}
(here the left tensor product is over the complex field and the right tensor product is over the reals). 

Throughout the paper we equate
$B(\CC^n, \CC^m)$ with the Hilbert space of $m \times n$ complex matrices $\M_{m,n}$ that 
comes naturally equipped with the Hilbert-Schmidt inner product 
\begin{equation*}
    \langle A, B \rangle_{\text{HS}} = \Tr A^{\dagger} B, 
\end{equation*}
for $A, B \in \M_{m,n}$. Analogously we equate $B(\CC^n)=\M_n$ and $B^{\text{sa}}(\CC^n)=\M^{\text{sa}}_n$ (note that $n^2$ is their complex and real dimension, respectively).
In this article we restrict our study of entanglement to bipartite Hilbert spaces $\CC^{d_1} \otimes \CC^{d_2}$
and can thus avail another standard identification, 
between operators in
$B(\CC^{d_1} \otimes \CC^{d_2} )$ and 
$d_1 d_2 \times d_1 d_2$ block matrices in $\M_{d_1 d_2}$.
Geometrically it is often convenient to drop the trace one constraint and (by abuse of notation) consider any positive semidefinite matrix as a state.
In view of the above identifications, when 
$\HH=\CC^{d_1} \otimes \CC^{d_2}$, 
the convex compact set $\D(\CC^{d_1} \otimes \CC^{d_2})$ is the base of the positive semidefinite cone $\PSD(\CC^{d_1} \otimes \CC^{d_2}) \subset \M_{d_1 d_2}^{\text{sa}}$, and similarly
$\Sep(\CC^{d_1} \otimes \CC^{d_2})$ is the base of the separable cone 
$\SEP(\CC^{d_1} \otimes \CC^{d_2}) \subset 
\PSD(\CC^{d_1} \otimes \CC^{d_2})$. 
The $\PSD$ cone is self-dual with respect to the Hilbert-Schmidt inner product. On the other hand,
it is straightforward to check that the dual of the separable cone consists of \textit{block-positive matrices}
\begin{equation*}
    \mathcal{BP}:= \set{ M \in \M_{d_1 d_2}^{\text{sa}} \colon 
    \sum_{j,k=1}^{d_1}  z_j \ols{z_k} M_{jk} \in \PSD(\CC^{d_2})
    \text{ for all } \texttt{z} \in \CC^{d_1}
    }.
\end{equation*}
By definition, the cone $\mathcal{PPT}$ and its trace 1 base PPT
consist of the states with positive partial transpose; in other words, $\rho \in \D$ is a \textit{PPT state} if 
$\Gamma(\rho) \succeq 0$, 
where $\Gamma:= T \otimes \Id$ is the partial transposition and $T$ is the transposition. We name 
a matrix \textit{decomposable} if it is a convex combination of
matrices in $\PSD$ and $\text{co-}\!\PSD:=\Gamma(\PSD)$. 
Using this notation, $\mathcal{PPT} = \text{co-}\!\PSD \cap \PSD$ 
and its dual cone is $\mathcal{PPT}^{\ast} = \text{co-}\!\PSD +\PSD$.
The inclusion $\SEP \subset \mathcal{PPT}$ is known as 
the \textit{Peres-Horodecki criterion} due to \cite{HorodMain} and \cite{Peres}, or
the \textit{PPT criterion} 
due to its construction.
The hierarchies of the above cones, their bases and their dual cones are summarised in \cref{tab:conesO}, where the cones are increasing from bottom to top whereas the dual cones are decreasing in the same order. 

\begin{table}[ht]
\label{tab:conesO}
\centering
 \begin{tabular}{||c | c | c || c||} 
 \hline
   \multicolumn{2}{|c|}{Cone of matrices $\mathcal{C}$ }& base of $\mathcal{C}$  & 
  dual cone $\mathcal{C}^{\ast}$ \\ 
 \hline\hline
 block positive & $\mathcal{BP}$ & BP & $\SEP$ \\ 
 \hline
 \!\!\! decomposable \!\!\! &\!\!\! $\text{co-}\!\PSD \!+\! \PSD$ \!\!\!& \!\!\!
 conv$\,\{\D \cup \Gamma(\D)\}$ \!\!\! &\!\!\! $\mathcal{PPT}$  \\
 \hline
 positive & $\PSD$ & $\D$ & $\PSD$ \\
 \hline
 PPT & $\mathcal{PPT}$ & PPT  \!\!\! &\!\!\!
 $\text{co-}\!\PSD \!+\! \PSD$\!\! \\
 \hline
 separable & $\SEP$ & $\Sep$ & $\mathcal{BP}$ \\ 
 \hline
\end{tabular}
\caption{Cones of operators on $\CC^{d_1} \otimes \CC^{d_2}$ represented as
cones of matrices in $\M^{\text{sa}}_{d_1 d_2}$.}
\end{table}

Next, we tackle more general maps acting between spaces of matrices 
(or operators).  
In order to emphasize the difference between such maps and the operators considered above, 
we call them 
\textit{quantum maps} or \textit{superoperators}. In what follows we consider
self-adjointness-preserving $\CC$-linear quantum maps
$\Phi \colon \M_{d_2} \rightarrow \M_{d_1}$
(note that every such map $\Phi$ 
is obtained from an $\RR$-linear map 
$\Psi \colon \M^{\text{sa}}_{d_2} \rightarrow \M^{\text{sa}}_{d_1}$
in a unique way by complexification and conversely, $\Psi$ is a restriction of $\Phi$). The map $\Phi$ is defined to be \textit{positive} if it maps positive semidefinite matrices to 
positive semidefinite matrices, and $\Phi$ is $n$-\textit{positive} if 
\begin{IEEEeqnarray*}{CCCCC}
\Phi \otimes \Id & \colon & 
\M^{\text{sa}}_{d_2 n}  &
\longrightarrow &
\M^{\text{sa}}_{d_1 n} 
\end{IEEEeqnarray*}
is positive (under the identifications 
$\M^{\text{sa}}_{d_i n} = B^{\text{sa}}(\CC^{d_i} \otimes \CC^n) = 
\M^{\text{sa}}_{d_i} \otimes \M^{\text{sa}}_{n}$). If $\Phi$ is $n$-positive for all $n \in \NN$, it is said to be \textit{completely positive}. 
For example, transposition is a positive map which is not 2-positive and is thus not completely positive. 
In quantum information theory, completely positive maps that are also trace-preserving are called \textit{quantum channels}. 
A completely positive map 
$\Phi \in B(\M_{d_2}^{\text{sa}}, \M_{d_1}^{\text{sa}})$ is \textit{entanglement breaking} if 
for any $n \in \NN$ and any $Z \in \M^{\text{sa}}_{d_2 n}$ the matrix
$\left( \Phi \otimes \Id_{\M_n} \right) Z$
is separable. 
From the definitions
it follows that the sets of entanglement breaking maps, completely positive maps and  positive maps are all convex cones living inside
the real vector space of self-adjointness-preserving quantum maps; we denote them as 
$\boldsymbol{EB}(\M_{d_2}^{\text{sa}}, \M_{d_1}^{\text{sa}}) \subset
\CP(\M_{d_2}^{\text{sa}}, \M_{d_1}^{\text{sa}}) \subset \bP(\M_{d_2}^{\text{sa}}, \M_{d_1}^{\text{sa}})$. 

The Choi isomorphism establishes a one-to-one correspondence between superoperators acting on matrix algebras and the Choi matrices acting on bipartite Hilbert spaces. 
In specified bases (denoted as $\{ e_i \}_{i=1}^d$  and  
$\{ E_{ij} \}_{i,j=1}^d$ for
$\CC^d$ and $\M_d$, respectively), the \textit{Choi isomorphism} is defined as
\begin{IEEEeqnarray*}{CCL}
\vspace{1.5mm}
B(\M_{d_2}, \M_{d_1}) &
\ \ \xrightarrow{C} \ \ &
\ \ \ \ \ \ \ \ \ \ \ \ \M_{d_1 d_2} \\ 
\Phi \colon \M_{d_2} \rightarrow \M_{d_1} & \mapsto & 
C(\Phi) \colon  \CC^{d_1} \otimes \CC^{d_2} \rightarrow 
\CC^{d_1} \otimes \CC^{d_2} \\
 & & \ \ \,  || \\
 & & \sum_{i,j} \Phi(E_{ij}) \otimes E_{ij}
\end{IEEEeqnarray*}
and $C(\Phi)$ is called the \textit{Choi matrix} of $\Phi$.
It is often convenient to rewrite the Choi matrix as follows:
\begin{equation} \label{eq:ChoiMat}
    C(\Phi) = \left( \Phi \otimes \Id_{\M_{d_2}} \right) (\ketbra{\chi}{\chi}),
\end{equation}
where 
$\chi = \sum_i e_i \otimes e_i \in \CC^{d_2}  \otimes  \CC^{d_2}$ is a \textit{maximally entangled vector}.
In quantum information theory it is standard to compute 
with the computational bases $\{\ket{i-1} \}_{i=1}^d$ of $\CC^d$ and  analogously $\{ \ket{ij}:=\ket{i} \otimes \ket{j} \}_
{\substack{0\leq i < d_1,\,  0\leq j < d_2}}$
of $\CC^{d_1}  \otimes  \CC^{d_2}$.

The Choi isomorphism $C$ induces 
an isomorphism between the real vector spaces of self-adjointness-preserving linear maps 
$\left\{ \Phi \colon \M_{d_2}^{\text{sa}} 
\rightarrow \M_{d_1}^{\text{sa}} \right\}$ and self-adjoint matrices 
$\M_{d_1 d_2}^{\text{sa}}$.
In particular, by Choi's theorem~\cite{ChoiCP}, $C$ induces an isomorphism between the self-dual cones $\CP(\M_{d_2}^{\text{sa}}, \M_{d_1}^{\text{sa}})$ and    
$\PSD(\CC^{d_1} \otimes \CC^{d_2})$. 
Indeed,  Choi's theorem states that $\Phi$ being completely positive is equivalent to $C(\Phi)$ being positive semidefinite, which in turn is equivalent to  
$\Phi$ having a particular structure 
\begin{equation} \label{eq:Kraus}
    \Phi(Z) = \sum_{k=1}^K A_k Z A^{\dagger}_k, 
\end{equation}
where $ A_1, \ldots, A_K \in \M_{d_1, d_2}$,
called the \textit{Kraus decomposition} of $\Phi$. 
Consequently the Choi isomorphism induces one-to-one correspondences between the cones of positive / decomposable / completely positive / PPT-inducing / entanglement breaking maps  and the cones of block positive / decomposable / positive semidefinite / PPT / separable  
matrices, respectively. 
We refer to the book of Aubrun and Szarek \cite{AandBbook} for an excellent exposition of the Choi isomorphism (and the related Jamio{\l}kowski isomorphism) and the correspondences it induces between the respective cones. 
The relations, following the notation in \cite{AandBbook}, are summarised in 
\cref{tab:conesSoO}. We remark that, by St{\o}rmer-Woronowitz theorem \cite{StorAold,Woron}, every state on $\CC^2 \otimes \CC^2$
or $\CC^2 \otimes \CC^3$ with positive partial transpose is separable; in other words, 
$\SEP(\CC^{d_1} \otimes \CC^{d_2})=\mathcal{PPT}(\CC^{d_1} \otimes \CC^{d_2})$ for $d_1 d_2 \leq 6$, which is equivalent to $\boldsymbol{P}(\M_{d_2}^{\text{sa}}, \M_{d_1}^{\text{sa}})=
\boldsymbol{DEC}(\M_{d_2}^{\text{sa}}, \M_{d_1}^{\text{sa}})$.

\begin{table}[ht]
\label{tab:conesSoO}
\centering
 \begin{tabular}{||c | c || c | c ||} 
 \hline
  \multicolumn{2}{||c||}{Cone of superoperators  $\boldsymbol{C}$} & 
  \multicolumn{2}{c||}{Cone of matrices $\mathcal{C}$ }   \\ 
 \hline\hline 
  & & & \\
 positive & $\boldsymbol{P}$ &
 block positive & $\mathcal{BP}$ \\ 
  & $\cup$ & & $\cup$ \\
 decomposable & $\boldsymbol{DEC}$ & 
 decomposable & $\text{co-}\!\PSD + \PSD$      \\
 & $\cup$ & & $\cup$ \\
 completely positive & $\boldsymbol{CP}$  & 
 positive semidefinite & $\PSD$ \\
& $\cup$ & & $\cup$ \\
 PPT-inducing & $\boldsymbol{PPT}$ & 
 "PPT"  & $\mathcal{PPT}$  \\
 & $\cup$ & & $\cup$ \\
 entanglement breaking & $\boldsymbol{EB}$ & 
 separable & $\SEP$ \\ 
  & & & \\
 \hline
\end{tabular}
\caption{Each cone $\boldsymbol{C}$ is connected to a cone 
$\mathcal{C}$ by the Choi isomorphism: 
$\Phi \in \boldsymbol{C} \Longleftrightarrow C(\Phi) \in \mathcal{C}$.}
\end{table}

The goal of the article is to 
construct new entanglement witnesses that certify entanglement in bipartite states. 
To this end we deploy the identifications in \cref{tab:conesO,tab:conesSoO} of the dual cone $\SEP^{\ast}$ with the cone of block positive matrices $\mathcal{BP}$ and the isomorphic cone of positive maps $\bP$. These identifications yield the following equivalent conditions for a state $\rho$ on $\HH=\CC^{d_1} \otimes \CC^{d_2}$:
\begin{enumerate}
    \item state $\rho$ is entangled,
    \item there exists $\sigma \in \SEP^{\ast}= \mathcal{BP}$
    such that $\langle \sigma, \rho \rangle_{\text{HS}}= 
    \Tr(\sigma \rho) <0$,
    \item there exists a positive map 
    $\Psi \colon \M_{d_2}^{\text{sa}} \rightarrow \M_{d_1}^{\text{sa}}$ 
    such that $\Tr\left(C(\Psi) \rho \right)<0$.
\end{enumerate}
The following Horodecki's entanglement witness theorem is a direct corollary of the above, where the positive map $\Phi=\Psi^{\dagger}$ is the adjoint of $\Psi$ in statement 
3.
\begin{corollary}[Horodecki's entanglement witness 
theorem \cite{HorodMain}]
\label{cor:HEW} Consider
$\HH = \CC^{d_1} \otimes \CC^{d_2}$. 
A state  $\rho$ on $\HH$ is entangled
if and only if there exists a positive map 
$\Phi \colon \M_{d_1}^{\text{sa}} \rightarrow \M_{d_2}^{\text{sa}}$
such that the matrix
$\left( \Phi \otimes \Id_{\M_{d_2}^{\text{sa}} }\right)\rho$ is not positive semidefinite.
\end{corollary}
Using the above notation, the positive map $\Phi \in \bP \backslash \CP$ or the block positive matrix 
$\sigma = C(\Psi) \in \mathcal{BP} \backslash \PSD$ 
(or the linear functional 
$\langle \sigma, \cdot \rangle_{\text{HS}}$) is defined to be
an \textit{entanglement witness}, because it witnesses the entanglement of the state $\rho$.
An entanglement witness $\sigma \in \mathcal{BP}$
is said to be an \textit{optimal entanglement witness} if the set 
$\textsc{E}(\sigma):= \set{\rho \in \D \colon \Tr(\sigma \rho) <0}$
is maximal, i.e., there exists no $\sigma' \in \mathcal{BP}$ such that $\textsc{E}(\sigma')$ would  strictly contain  
$\textsc{E}(\sigma)$. It follows from the $S$-lemma (see \cite{Slemma} and \cite[Appx. \hspace{-3pt}C and \S 2.4]{AandBbook}) that when $\sigma$ 
lies on an extreme ray of $\mathcal{BP}$ and $\sigma \notin \PSD$, then it is an optimal entanglement witness.

\subsection{Historical development and Contributions} 
In~\cite{Klep,SzarEtAl} it was shown that, in the general setting, the cone of positive maps is much bigger than the cone of completely positive maps. 
By comparing the relations between the cones in \cref{tab:conesSoO,tab:conesO}, this is equivalent to the statement that the cone of positive semidefinite matrices is much bigger than the cone of separable matrices. 
Consequently, the compact convex set $\D (\CC^{d_1} \otimes \CC^{d_2})$ is a larger set than 
$\Sep (\CC^{d_1} \otimes \CC^{d_2})$
despite having the same inradius by the Gurvits-Barnum~\cite{GurBar} theorem
(with respect to the Hilbert-Schmidt norm and center in the maximally mixed state $\frac{1}{d_1 d_2} I$), and thus 
having the same dimension
in $\RR$,
\begin{displaymath}
  \dim \Sep = \dim \D = (d_1 d_2)^2 -1.
\end{displaymath}
Moreover, there exists no known simple description of the facial structure of the convex set $\Sep$ (unlike the well understood geometry of the set of states $\D$, see for example \cite{BengZy}). This is related to the fact that the fundamental question in quantum physics, quantum information theory (QIT) and
quantum communication - whether a state is separable or entangled - is
an NP-hard problem \cite{Gha}. 

This paper aims to shed some light on the geometry of 
$\SEP (\CC^{3} \otimes \CC^{3})$, the cone of separable states on 
$\CC^3 \otimes \CC^3$. 
The main contribution of the article is in presenting the families of positive maps on $\M_3^{\text{sa}}$
from \cite{BucSiv} as optimal entanglement witnesses and consequently 
give an explicit construction of new entangled states. 
(By Choi's isomorphism, this is equivalent to the construction of
witnesses of non-entanglement breaking maps. Indeed, from \cref{tab:conesSoO} we read that, an entangled state can be represented as the Choi matrix of a completely positive map that is not entanglement breaking.) 
Positive maps that are not completely positive can be used to prove (witness) that certain mixed states are entangled - therefore the name entanglement witnesses.
We design a semidefinite program  in such a way that it outputs entangled states which are not detected by other known positive maps, specifically the transposition and the Choi map. 
The advantage of our symbolic computation of entanglement witnesses by 
the method of "prescribing zeros" \cite{BucSiv}, is that we simultaneously
obtain multidimensional families of entanglement witnesses; as a demonstration of method's strength, we obtain a 5-parameter family of positive maps that amalgamates all the generalizations of Choi's map in the literature. 
Moreover, our entanglement witnesses lie on the extremal rays of the respective cones, thus they are optimal entanglement witnesses. Such explicit families of entanglement witnesses and the corresponding entangled states provide a tighter approximation of the separable cone $\SEP (\CC^{3} \otimes \CC^{3})$ (see \cref{fig:Psi_t}). 

% The outline is not required, but we show an example here.
The paper is organized as follows. We start by setting the stage: we introduce the notion of entanglement witnesses and
the advantages of them being extremal and indecomposable.
The classical example of the Choi map is recalled, which is arguably the second most famous positive map after the transposition.
\Cref{sec:main} reviews the method from \cite{BucSiv}, of constructing new families of positive maps
on $\M_3^{\text{sa}}$, and explains their relevance and 
implications in quantum theory. 
In \cref{sec:examples} we extend the main theorems from  \cite{BucSiv} by constructing families of positive maps that are neither completely positive nor co-completely positive. This means that such maps are capable of witnessing entanglement of states which the transposition map fails to certify.
We also address how to make entanglement witnesses unital (or trace preserving).
Our algorithm, phrased as a semidefinite program, is provided in \cref{sec:alg} - it involves the new maps and outputs examples of entangled states on  
$\CC^3 \otimes \CC^3$ whose entanglement can neither be witnessed by the transposition map nor by the Choi map.
Experimental results, both some exact and some numerical, are in \cref{sec:experiments}, and the conclusions follow in
\cref{sec:conclusions}.

In mathemathics, positive maps are key notions in
$C^{\ast}$-algebras (see \cite{StorA}), in the study of nonnegative polynomials and sums of squares \cite{Bleck4,Bleck5},  and in connection with cone optimization, where positive maps act between real symmetric matrices (following the ideas of Lasserre and Parrilo, 
polynomial optimization problems can be relaxed and transferred into semidefinite optimization problems \cite{NieZh,Bleck2}). 
It is natural to extend positive maps from acting on $\M^{\text{sym}}_d(\RR)$ 
to positive maps on $\M^{\text{sa}}_d(\CC)$, and thus making them suitable 
for QIT.
However, in general extra effort is required to ensure that the complexification of an extremal positive map is also extremal \cite{KCHa}.
The relevance of positive maps in QIT was observed by the Horodecki group \cite{HorodMain}, particularly in connection with the entanglement detection, 
and since then there are many examples of positive maps that are not 
completely positive in the literature 
(for an overview  we refer the interested 
reader to the survey paper \cite{KyeExp} and the references in \cite{BucSiv}). 
By St{\o}rmer-Woronowitz theorem \cite{StorAold,Woron}, $d_1=d_2=3$ is the smallest dimension where the transposition map does not detect all entangled states. 
Extensive study of the geometry of entangled bi-qutrit PPT states, their classification and construction of certain simplicial faces of 
$\Sep (\CC^{3} \otimes \CC^{3})$ are due to Kye et al. 
(see \cite{KyePPT1,KyePPT2} and the parametrized examples of entangled PPT states in their preceding works). In \cite{Hermes,ChrConj} the authors further explore the relations between entanglement and (complete/co-complete) positivity. 
Like real algebraic geometry, quantum information theory avails various forms of semidefinite programming \cite{Wat}. For example,
semidefinite programming is a well-known modus operandi for
testing entanglement~\cite{SDPentang}.

\section{Entanglement witnesses}
\label{sec:main}

In this section we motivate the study of entanglement witnesses with the two most known examples, the transposition and the Choi map.
We describe our method of construction of optimal entanglement witnesses. 

\subsection{Setting the stage}

Note that Horodecki’s  entanglement  witness  theoorem in \cref{cor:HEW}, for $\Phi=T$, yields the  PPT criterion or the Peres-Horodecki criterion:
\begin{equation}
\label{eq:PPTcriterion}
    \SEP \subset \mathcal{PPT} := \PSD
    \cap\, \Gamma (\PSD),
\end{equation}
where $\Gamma := T \otimes \Id$.
The strength of the PPT criterion is in detecting entanglement: if the partial transpose of a state is not positive, the state itself must be non-separable, i.e., entangled. 
On the other hand, the transposition fails to detect the states in the nonempty set $\PPT \backslash \Sep$ as illustrated in \cref{fig:PT}.
This, and the fact that the partial transposition detects entanglement in any pure state, makes the transposition  the most fundamental entanglement witness in quantum information theory. 
PPT states have attracted much attention in the literature as they can be seen as rough approximations to separable states (see 
\cref{eq:PPTcriterion} and \cref{fig:PT}). 
\begin{figure}[htbp]
  \centering
  \includegraphics[width=0.73\textwidth, height=95mm]{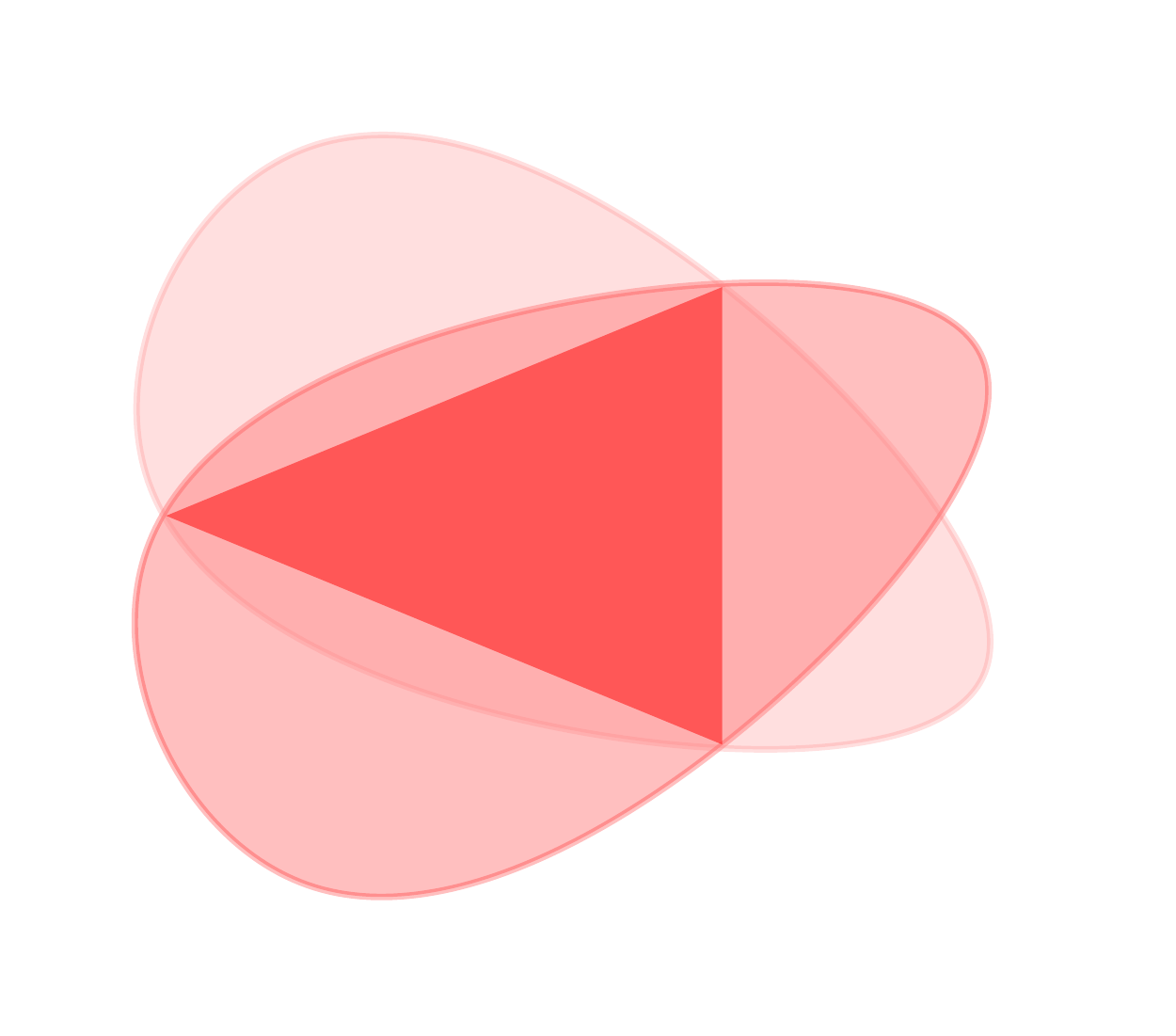}
   \put(-190,205){$\Gamma(\D)$}
    \put(-155,134){$\Sep$ }
    \put(-190,70){$\D$}
     \put(-60,197) {\vector(0,-1){25}}
            \put(-80,200){states detected by $T$}
        \put(-80,70) {\vector(0,1){60}}
       \put(-85,60){$\PPT = \D \cap \, \Gamma(\D)$}
         \put(-85,48){states not detected by $T$}
  \caption{ The inclusion 
    $\Sep \left( \CC^3 \otimes \CC^3\right) \subset 
    \PPT \left( \CC^3 \otimes \CC^3\right)$ is strict.}
  \label{fig:PT}
\end{figure}
By applying
transposition together with other entanglement witnesses (i.e., positive maps that are not completely positive) on the cone $\PSD$, we can get better approximations of the separable cone. Indeed, by \cref{cor:HEW} the following inclusion holds for a set 
$\mathbf{A} \subset \bP \backslash \boldsymbol{CP}$:
\begin{equation*}
    \SEP \subset 
   \PSD \bigcap_{\Psi \in \mathbf{A} } 
    \left( \Psi^{\dagger} \otimes \Id \right) (\PSD);
\end{equation*}
note that for a richer family of positive maps $\mathbf{A}$ we get tighter approximation of $\SEP$.
(Alternatively, the above formula can be stated as 
$ \SEP \subset 
   \bigcap_{\Psi \in \mathbf{B} } 
    \left( \Psi^{\dagger} \otimes \Id \right) (\PSD)$, 
where we can assume that $\mathbf{B} \subset \bP$ contains $\Id$.)

The first examples of positive maps, which are not completely positive when restricted to real symmetric matrices (i.e., indecomposable positive maps), are due to Choi~\cite{Choi72,ChoiSD}. 
In \cite{KCHa} 
the Choi map 
$\Psi_C \colon \M_3^{\text{sa}} \to \M_3^{\text{sa}}$ is defined as 
\begin{equation}
\label{eq:ChoisMap}
    \Psi_C \left( 
    \begin{bmatrix}
    z_{00}& z_{01}&z_{02} \\
z_{10}&   z_{11} &z_{12}\\
z_{20} &z_{21} &  z_{22}
    \end{bmatrix}
    \right) = 
     \begin{bmatrix}
    z_{00}+z_{11}& -z_{01}&-z_{02} \\
-z_{10}&   z_{11}+z_{22} &-z_{12}\\
-z_{20} &-z_{21} &  z_{00}+z_{22}
    \end{bmatrix}.
\end{equation}
The Choi matrix of $\Psi_C$ is not positive semidefinite
and it lies on an extreme ray of block positive matrices by \cite{KCHa}, therefore  $\Psi_C$ is an optimal entanglement witness.  
Note that $\Psi_C$ is neither stronger nor weaker entanglement witness than the PPT criterion, as illustrated in \cref{fig:PT_Choi}. 
For example, 
it is straightforward to check that the mixed state 
\begin{IEEEeqnarray*}{C} 
  \rho= 
  \frac{1}{21}   \left[\begin{array}{ccc|ccc|ccc}
    2 & \cdot & \cdot & 
    \cdot & 2 & \cdot &
    \cdot & \cdot & 2   \\
 \cdot & 1 & \cdot & \cdot & \cdot & \cdot & \cdot & \cdot & \cdot   \\
 \cdot & \cdot & 4 & \cdot & \cdot & \cdot & \cdot & \cdot & \cdot  \\
 \hline 
\cdot & \cdot & \cdot &  4 & \cdot & \cdot & \cdot & \cdot & \cdot   \\
 2 & \cdot & \cdot & \cdot & 2 & \cdot & \cdot & \cdot & 2   \\
 \cdot & \cdot & \cdot & \cdot & \cdot & 1 & \cdot & \cdot & \cdot  \\
 \hline 
\cdot & \cdot & \cdot &  \cdot & \cdot & \cdot & 1 & \cdot & \cdot   \\
 \cdot & \cdot & \cdot & \cdot & \cdot & \cdot & \cdot & 4 & \cdot   \\
 2 & \cdot & \cdot & \cdot & 2 & \cdot & \cdot & \cdot & 2
\end{array}\right]
\end{IEEEeqnarray*}
from \cite{StorA} has positive partial transpose,
however
$(\Psi_C^{\dagger} \otimes \Id )\, \rho$ has an eigenvalue equal to $-\frac{1}{21}$.
On the other hand, the pure state $\ketbra{v}{v}$ with  $v=(1,0,0,0,0,0,0,0,1)$ (or $v=\ket{00}+\ket{22}$ in the computational basis) is detected by $\Gamma = T \otimes \Id$ and not by 
$\Psi_C$. 

\begin{figure}[htbp]
  \centering
  \includegraphics[width=0.73\textwidth, height=95mm]{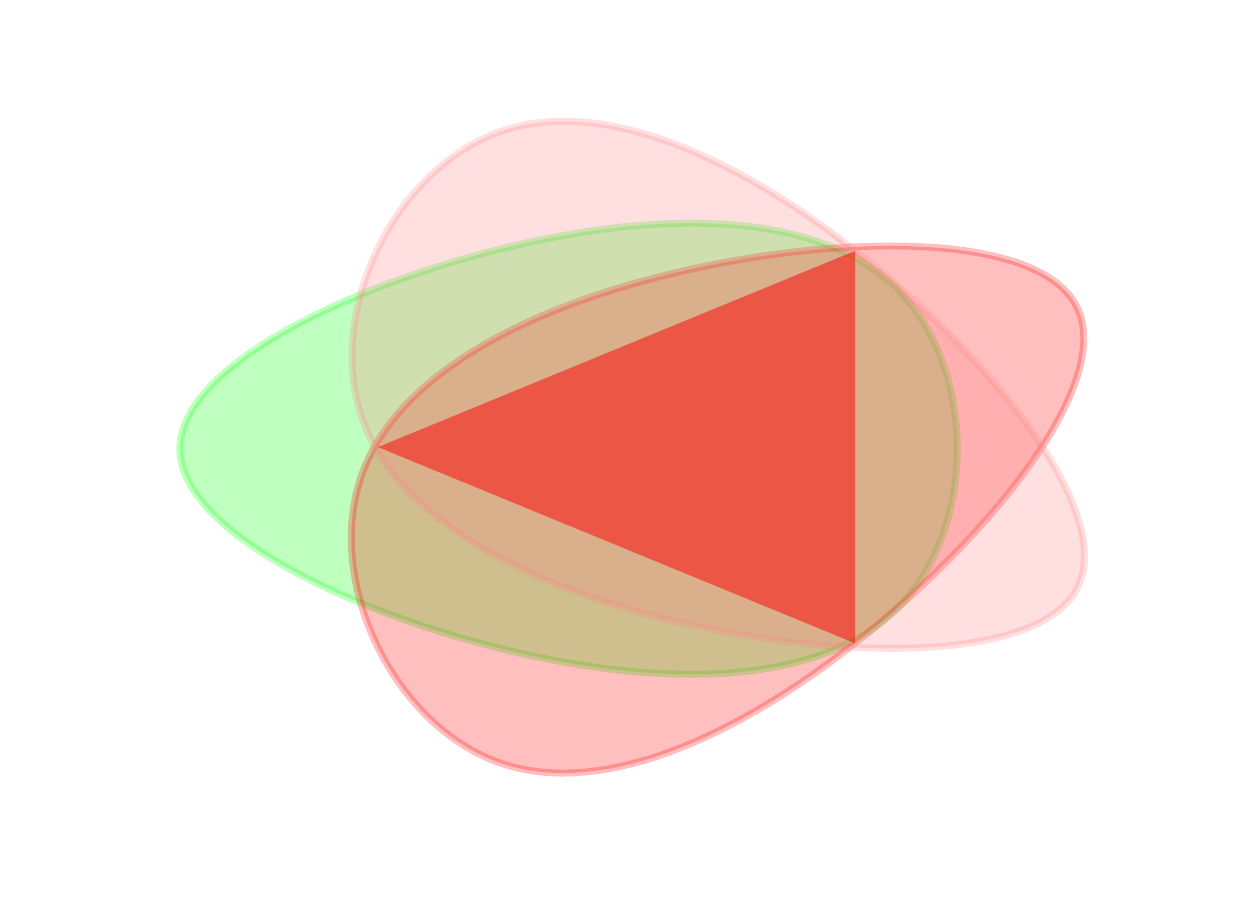}
   \put(-160,215){$\Gamma(\D)$}
      \put(-203,81) {\vector(1,0){53}}
      \put(-295,80){states detected by $T$ }
       \put(-295,70){and not by $\Psi_C$}
    \put(-135,134){$\Sep$ }
    \put(-250,134){$( \Psi_C^{\dagger} \otimes \Id )(\D)$ }
    \put(-160,50){$\D$}
         \put(-75,68) {\vector(0,1){63}}
         \put(-82,61) {\vector(-1,1){23}}
       \put(-80,58){$\D \cap \, \Gamma(\D) 
       \cap (\Psi_C^{\dagger} \otimes \Id)(\D) \, \backslash \Sep$}
       \put(-80,46){entangled states }
       \put(-80,36){detected by neither $\Psi_C$ nor $T$}
       \put(-56,201) {\vector(0,-1){67}}
            \put(-80,215){states detected by $\Psi_C$ }
             \put(-80,205){and not by $T$}
  \caption{Additional constrain $\left( \Psi_C^{\dagger} \otimes \Id \right) \rho \succeq 0$, where $\Psi_C$ is the Choi map.}
  \label{fig:PT_Choi}
\end{figure}

\begin{remark} \label{rem:Rzeros}
Actually, in~\cite{ChoiSD} Choi considered the following correspondence between 
linear maps on real symmetric matrices and biquadratic forms in variables 
$\x=(x_0, x_1, x_2)$, $ \y=(y_0, y_1, y_2)$:
\begin{equation} \label{eq:biqForm}
 \Psi \colon \M_3^{\text{sym}} \rightarrow \M_3^{\text{sym}}
  \  \invertleadsto	\! \leadsto \ p_{\Psi}(\x, \y) := \mel{\y}{\, 
  \Psi\left(
  \ketbra{\x}{\x} \right) \,}{\y},  \ \x, \y \in \RR^3,
\end{equation}
and proved that it induces an isompophism
between linear maps / positive maps / completely positive maps on symmetric matrices and  biquadratic forms / nonnegative biquadratic forms / sums of squares (SOS) of bilinear forms on $\RR^3 \times \RR^3$, respectively.
Consequently, in the $\RR$-setting, positive maps can be constructed by specifying the sets of zeros of nonnegative biquadratic forms (see \cite{Rezn,Quarez}). Quarez \cite{Quarez} proved that a 
nonnegative biquadratic form which is not a sum of squares can have at most 10 zeros. Contrarily, the number of real zeros of an SOS form is either infinite or at most 6. This implies that nonnegative biquadratic forms with 7, 8, 9 or 10 zeros define positive maps that are not completely positive. 
For example,  the Choi map \eqref{eq:ChoisMap} restricted to $\M_3^{\text{sym}}$ induces the 
biquadratic form 
\begin{equation*}
    p_{\Psi_C} = x_0^2 y_0^2 + x_1^2 y_1^2 + x_2^2 y_2^2 +
    x_0^2 y_2^2 + x_1^2 y_0^2 + x_2^2 y_1^2
    - 2 x_0 x_1 y_0 y_1 - 2 x_0 x_2 y_0 y_2 - 2 x_1 x_2 y_1 y_2
\end{equation*}
with 7 zeros 
\begin{multline} \label{eq:ChoiZeros}
      (1,1,1; 1,1,1), (1,1,-1;1,1,-1), (1,-1,1;1,-1,1), (-1,1,1;-1,1,1), \\
(1,0,0;0,1,0), (0,1,0,0,0,1), (0,0,1;1,0,0).
\end{multline}
This gives an alternative proof that the Choi map on $\M_3^{\text{sa}}$ is not completely positive (since its restriction to $\M_3^{\text{sym}}$ is not completely positive).

On the other hand,
the extensions of extremal positive maps from $\M_3^{\text{sym}}$ to $\M_3^{\text{sa}}$ are in general not extremal.
Nearly 40 years after Choi's construction, K.-C. Ha~\cite{KCHa} gave an elaborate proof that, also in the complex setting, the Choi map \eqref{eq:ChoisMap} lies on an extreme ray of the cone of positive maps 
$\bP(\M_3^{\text{sa}}, \M_3^{\text{sa}})$. 
In \cref{ex:Choi} we will show 
how the extremality of the Choi map follows directly from our construction.

The geometry of nonnegative polynomials is
relevant in all areas of polynomial optimisation, see~\cite{Bleck2}. In higher dimensions, most of recent developments in polynomial optimization are based on sum of squares decomposition of nonnegative polynomials.
We believe that the methods and the restrictions of positive maps considered in this paper
could contribute to the understanding of symmetric nonnegative forms.  
\end{remark}

\subsection{Construction of positive maps: the method of prescribing zeros} \label{subsec:construction}

In \cref{subsec:Notation} we explained that having an entanglement witness is equivalent to having a positive map $\Psi$ that is not completely positive, which in turn is equivalent to
the associated Choi matrix $C(\Psi)$ being block positive but not positive semidefinite. Moreover, $\Psi$ is an optimal entanglement witness if and only if it
belongs to an extreme ray in the positive cone $\bP$.

In what follows we describe our construction of families of entanglement witnesses, 
which we call "the method of prescribing the zero set", first discovered in \cite{BucSiv}.
We construct positive maps by extending Choi's approach, described in \eqref{eq:biqForm}, from the real symmetric
matrices to Hermitian matrices.
To a given self-adjointness preserving map
$\Psi \colon \M_3^{\text{sa}} \rightarrow  \M_3^{\text{sa}}$ 
we associate the form  
\begin{equation*}
  p_{\Psi}(\x, \y) := \mel{\y}{\, \Psi\left(
  \ketbra{\x}{\x} \right) \,}{\y},
\end{equation*}
where $\x=(x_0,x_1,x_2) \in \PP^2(\CC)$ and $\y=(y_0,y_1,y_2) \in \PP^2(\CC)$ (here we view $\x$ and $\y$ as points 
in the projective plane and not in $\CC^3$, 
which will be convenient when we consider the set of zeros of $p_{\Psi}$). 
Then, it follows directly from the definition of positive maps that,
$\Psi$ is positive  if and only if the associated form $p_{\Psi}$ is nonnegative (or equivalently, written in real variables, $p_{\Psi}$ is a nonnegative polynomial in
$ \frac{1}{2} (x_k + \ols{x_k}),\  \frac{1}{2 i} (x_k - \ols{x_k})$ 
and 
$\frac{1}{2} (y_k + \ols{y_k}),\  \frac{1}{2 i} (x_y - \ols{y_k})$ ).
The adjoint of a positive map
$\Psi \colon \M_3^{\text{sa}} \rightarrow \M_3^{\text{sa}}$  is another  positive map, which can be obtained from the following equalities 
\begin{IEEEeqnarray*}{CCCCC}
 \langle \Psi(\ketbra{\x}{\x}), \ketbra{\y}{\y} \rangle_{\text{HS}} &=&
 \Tr \left( \Psi(\ketbra{\x}{\x})  \ketbra{\y}{\y} \right) & = & p_{\Psi}(\x, \y) \\
  || & & & & \\
  \langle \ketbra{\x}{\x}, \Psi^{\dagger}(\ketbra{\y}{\y}) 
  &=&
 \Tr \left( \ketbra{\x}{\x}  \Psi^{\dagger}(\ketbra{\y}{\y}) \right) & = & p_{\Psi^{\dagger}}(\y, \x).
\end{IEEEeqnarray*}
The above equalities together with the Kraus decomposition \eqref{eq:Kraus} also imply that $\Psi$ is completely positive if and only if $\Psi^{\dagger}$ is completely positive.

When constructing a $\CC$-linear map $\Psi \colon \M_3 \rightarrow  \M_3$ it is advantageous to write it in coordinates, like
\begin{equation} \label{eq:PsiCoordinates}
    \Psi \left( \left[ z_{ij} \right]_{i,j=0}^2 \right)= 
\left[ \sum_{k,l=0}^2 a_{ijkl} \, z_{kl} \right]_{i,j=0}^2 
\ \text{ for } a_{ijkl} \in \CC.
\end{equation}
Note that $\Psi$ is 
self-adjointness preserving if and only if $\Psi(Z^{\dagger})=\Psi(Z)^{\dagger}$ for all $Z \in \M_3$, which in turn is equivalent to the Choi matrix $C(\Psi)$ being Hermitian. This shows 
that the coefficients of $\Psi$ 
are related as 
\begin{equation} \label{eq:SAcoeff}
    a_{ijkl}=\ols{a_{jilk}} \ \text{ for all } \ 
    0\leq i, j, k, l \leq 2.
\end{equation}
In the same coordinates, the adjoint map is
$ \Psi^{\dagger} \left( \left[ z_{ij} \right]_{i,j=0}^2 \right)= 
\left[ \sum_{k,l=0}^2 a_{lkji} \, z_{kl} \right]_{i,j=0}^2$.
Next we fix a set of points $\mathcal{Z} \subset \PP^2(\CC) \times \PP^2(\CC)$ and determine necessary conditions on $a_{ijkl}$, so that $\Psi$ is positive and $p_{\Psi}(\x,\y)=0$ for all $(\x,\y) \in \mathcal{Z}$. 
Under the assumption of positivity, both matrices $\Psi(\ketbra{\x}{\x})$ and $\Psi^{\dagger}(\ketbra{\y}{\y})$ are positive semidefinite, therefore
each zero $(\x,\y) \in \mathcal{Z}$ of $p_{\Psi}$ imposes additional linear relations among the coefficients $a_{ijkl}$ in the following way: 
\begin{equation} \label{eq:zeroCoeff}
    \Psi(\ketbra{\x}{\x})\ket{\y}=
    \begin{bmatrix}
        0 \\ 0 \\ 0
    \end{bmatrix}
   \  \text{ and } \
    \Psi^{\dagger}(\ketbra{\y}{\y}) \ket{\x}=
     \begin{bmatrix}
        0 \\ 0 \\ 0
    \end{bmatrix}.
\end{equation}

The group $\text{PGL}_3 \times \text{PGL}_3$ acts naturally on the cone of positive maps $\bP(\M_3^{\text{sa}}, \M_3^{\text{sa}})$ by
\begin{equation} \label{eq:PQaction}
     \Psi(Z)  \ \mapsto \ 
   Q^{\dagger}\, \Psi\left(P Z P^{\dagger} \right) Q,
\end{equation} where $P, Q \in \text{PGL}_3$, which is eqivalent to the change of variables in biquadratic forms
\begin{equation*}
     \mel{\y}{\, \Psi\left(
  \ketbra{\x}{\x} \right) \,}{\y} \   \mapsto \ 
  \mel{Q \y}{\, \Psi\left(
  \ketbra{P \x}{P \x} \right) \,}{ Q \y}.
\end{equation*} 
The above action also preserves the complete positivity, which can be seen directly from the Kraus decomposition \eqref{eq:Kraus}. 
Moreover, the cardinality of 
the zero set of $p_{\Psi}$ (i.e., the set 
$\set{ (\x, \y) \in \PP^2(\CC) \times \PP^2(\CC) 
\colon p_{\Psi}(\x, \y)=0 }$) is preserved,
and when the cardinality is at least four, we can by \cite{BucSiv} assume that four of
the zeros equal
\begin{equation} \label{eq:4zeros}
   (1,1,-1;1,1,-1), (1,-1,1;1,-1,1), (-1,1,1;-1,1,1) \text{ and }
    (1,1,1; 1,1,1).
\end{equation} 

We are now ready to summarize our construction of positive maps $\Psi$ by prescribing a sufficiently big set
$\mathcal{Z} \subset \PP^2(\CC) \times \PP^2(\CC)$ as zeros of 
the nonnegative forms $p_{\Psi}$. 
Without loss of generality we can assume that $\mathcal{Z}$ 
contains the four points in \eqref{eq:4zeros}.
We start by  solving the linear system of equations \eqref{eq:SAcoeff} and \eqref{eq:zeroCoeff} for all $(\x,\y)\in \mathcal{Z}$; the family of solutions $\set{ \Psi }$ satisfies necessary conditions for $\Psi$ being positive and $p_{\Psi}$ having the prescribed zeros. Next we determine which of these candidate maps are  
\begin{itemize}
    \item positive, i.e., $\Psi \left( \ketbra{\x}{\x} \right) \succeq 0$ for all $\x \in \PP^2(\CC)$;
    \item completely positive, i.e., the Choi matrix $C(\Psi) \succeq 0$;
    \item extremal, i.e., $\Psi$ is not a sum of different nonzero positive maps. 
\end{itemize}
In general the most difficult to verify is the positivity  (recall the top rows in \cref{tab:conesO,tab:conesSoO} showing that the positivity problem is as hard as the NP-hard separability problem). 
We are able to directly condition
which of the maps are positive. (Our approach differs from \cite{NieZh}, where the authors construct a hierarchy of semidefinite relaxations for minimizing the associated biquadratic form over the unit spheres. 
Alternative approach \cite{CueNet,NetPla} relates positivity 
with linear matrix inequalities that define spectrahedra and their shadows.) 

On the other hand, it is straightforward to determine which of the constructed positive maps
are extremal: if $\Psi = \Psi_1 + \Psi_2$, then $p_{\Psi} = 
p_{\Psi_1} + p_{\Psi_2}$, and both sets of zeros of $p_{\Psi_1}$
and $p_{\Psi_2}$ must contain the set $\mathcal{Z}$; this means that 
$\Psi_1, \Psi_2$ and $\Psi$
belong to the same family of positive maps. 

In the following example we construct a 5-parameter family of positive maps that includes the Choi map \eqref{eq:ChoisMap} and its  generalizations known in the literature. 

\begin{exmp}  \label{ex:Choi}
Consider the 3-parameter set of points in $\PP^2(\CC) \times \PP^2(\CC)$
\begin{multline} \label{eq:zerosChoiC}
\mathcal{Z}_C = 
 \left\{     (e^{i \varphi_0},e^{i \varphi_1},e^{i \varphi_2};\, 
      e^{i \varphi_0},e^{i \varphi_1},e^{i \varphi_2}), \right.  \\
\left. (1,0,0;\, 0,1,0), (0,1,0;\, 0,0,1), (0,0,1;\, 1,0,0)
 \right\},
\end{multline}
where $\varphi_0,\varphi_1,\varphi_2 \in [0, 2 \pi)$.  
It is straightforward to check that the points in $\mathcal{Z}_C$ 
are the zeros of 
$p_{\Psi_C}(\x, \y):=  \mel{\y}{\, \Psi_C\left(
  \ketbra{\x}{\x} \right) \,}{\y}$,
where $\Psi_C  \colon \M_3^{\text{sa}} \rightarrow \M_3^{\text{sa}}$ is the Choi map defined in \eqref{eq:ChoisMap}. (For this reason we will call 
$\mathcal{Z}_C$ the \textit{Choi set of zeros}.) 
Note that exactly 7 of the points in $\mathcal{Z}_C$ have coordinates in $\RR$ and they are equal to \eqref{eq:ChoiZeros}. 

We will construct all positive maps $\Psi$ such that the zero set of the associated forms $p_{\Psi}$ contains $\mathcal{Z}_C$. 
If we write $\Psi$ in coordinates \eqref{eq:PsiCoordinates}, then its coefficients $a_{ijkl} \in \CC $ need to satisfy the relations in \eqref{eq:SAcoeff} and \eqref{eq:zeroCoeff} for every $(\x, \y) \in \mathcal{Z}_C$. We obtain a linear system of equations in $a_{ijkl} := 
r_{ijkl} + i\, c_{ijkl}$ with the following 5-dimensional family of solutions:
\begin{equation*} 
    \Psi \left( \left[ z_{ij} \right]_{i,j=0}^2 \right) =
    \begin{bmatrix}
r_0 z_{00} + r_{0011} z_{11}  & 
(r_{0101}+ i\, c_{0101}) z_{01} & 
(r_{0202} - i\, c_{0101}) z_{02} \\
 (r_{0101} - i\, c_{0101}) z_{10} & 
 r_1 z_{11} + r_{0011} z_{22}  & 
 (r_{1212}+ i\, c_{0101} ) z_{12} \\ 
   (r_{0202} + i\, c_{0101})  z_{20} & 
   (r_{1212} - i\, c_{0101} )  z_{21} &
  r_{0011} z_{00} +r_2 z_{22}
    \end{bmatrix},
\end{equation*}
where we denote $r_0:= -(r_{0011} \!+\! r_{0101} \!+\! r_{0202})$, 
$r_1:= -(r_{0011} \!+\! r_{0101} \!+\! r_{1212})$ and
$r_2:= -(r_{0011} \!+\! r_{0202} \!+\! r_{1212})$.

Next we analyse which of these solutions define positive maps $\Psi$, which is equivalent to $\Psi(\ketbra{\x}{\x})$ being positive semidefinite for all $\x \in \PP^2(\CC)$.
To this end we view the $1 \times 1$ and $2 \times 2$ principal minors and the determinant of $\Psi(\ketbra{\x}{\x})$ as linear, quadratic and cubic polynomials in $s_i:=x_i \ols{x_i}=|x_i|^2$, respectively. 
The minors attain a simpler form if we replace
\begin{IEEEeqnarray}{C} \label{eq:paramrepl}
r_{0011}  \rightarrow r, \ \  c_{0101} \rightarrow c, \\
r_{0101} \rightarrow -\frac{1}{2} (r \!+\! r_0 \!+\! r_1 \!-\! r_2),\,  
r_{0202} \rightarrow -\frac{1}{2} (r  \!+\! r_0  \!-\! r_1  \!+\! r_2),\,  r_{1212} \rightarrow -\frac{1}{2} (r  \!-\! r_0  \!+\! r_1  \!+\! r_2).
\nonumber
\end{IEEEeqnarray}
Then the following expressions must be nonnegative for all $s_0, s_1, s_2 \geq 0$: from the diagonal elements of $\Psi(\ketbra{\x}{\x})$ we get
\begin{equation} \label{eq:minors1}
    r_0\, s_0 + r\, s_1, \ \
    r_1\, s_1 + r\, s_2, \ \ 
     r\, s_0 +r_2\, s_2,
\end{equation}
the principal $2 \times 2$ minors are
\begin{IEEEeqnarray}{L} \label{eq:minors2}
\left( 
r_0 r_1 -1/4 (r + r_0 + r_1 - r_2)^2 - c^2
\right) s_0 s_1 +  r r_0\, s_0 s_2 + 
r r_1\, s_1^2 + r^2\, s_1 s_2, 
\vspace{1mm} \nonumber \\
r r_0\,  s_0^2 + r^2\, s_0 s_1  +
\left( 
 r_0 r_2 -1/4 (r + r_0 - r_1 + r_2)^2 - c^2
\right) s_0 s_2 + r r_2\, s_1 s_2 , \vspace{1mm} \\
r r_1\, s_0 s_1 + r^2\, s_0 s_2 + 
\left( r_1 r_2 - 1/4  (r - r_0 + r_1 + r_2)^2 -c^2 
\right) s_1 s_2 + 
     r r_2\, s_2^2, \nonumber 
\end{IEEEeqnarray}
and the determinant is
\begin{IEEEeqnarray}{C}  \label{eq:minors3}
r \left( 
    r_0 r_1 -1/4 (r + r_0 + r_1 - r_2)^2 - c^2 
\right) s_0^2 s_1 +  \nonumber  \\ 
 r \left( 
r_0 r_2 -1/4 (r + r_0 - r_1 + r_2)^2 - c^2
\right) s_0 s_2^2 +   \nonumber  \\
r \left( 
  r_1 r_2 -1/4 (r - r_0 + r_1 + r_2)^2 - c^2
    \right) s_1^2 s_2 + \\
  r^2 r_1\, s_0 s_1^2 +   r^2 r_0\, s_0^2 s_2  + 
  r^2 r_2\, s_1 s_2^2 -  \nonumber  \\
 \frac{r}{4}  \left(
   6(r_0 r_1  +  r_0 r_2 +  r_1 r_2 ) +
 2 r (r_0 + r_1 + r_2) - 3 (r^2 + r_0^2 +  r_1^2 +r_2^2)  - 12 c^2  
    \right) s_0 s_1 s_2.  \nonumber 
\end{IEEEeqnarray}

Additionally, from the Choi matrix we can determine whether  
$\Psi$ is completely positive. We evaluate the Choi matrix $C(\Psi)$ by \eqref{eq:ChoiMat},
\begin{equation*}
    \begin{bmatrix}
r_0 & \cdot & \cdot & \cdot & 
r_{0101} + i\, c & \cdot & \cdot & \cdot & r_{0202} - i\, c \\ 
\cdot & r & 
\cdot & \cdot & \cdot & \cdot & \cdot & \cdot & \cdot  \\ 
\cdot &\cdot &\cdot &\cdot &\cdot &\cdot &\cdot &\cdot &\cdot  \\ 
\cdot &\cdot &\cdot &\cdot &\cdot &\cdot &\cdot &\cdot &\cdot  \\ 
r_{0101} - i\, c &
 \cdot & \cdot & \cdot & r_1 & 
 \cdot & \cdot & \cdot & r_{1212} + i\, c \\ 
\cdot & \cdot & \cdot & \cdot & \cdot & r & 
\cdot & \cdot & \cdot \\ 
\cdot & \cdot & \cdot & \cdot & \cdot & \cdot & r & 
\cdot & \cdot  \\
\cdot &\cdot &\cdot &\cdot &\cdot &\cdot &\cdot &\cdot &\cdot  \\ 
r_{0202} +i\, c & \cdot & \cdot & \cdot & 
r_{1212} - i\, c & 
\cdot & \cdot & \cdot & r_2
    \end{bmatrix}
\end{equation*}
(for clarity purposes
we write only the nonzero entries, and $r_{0101}, r_{0202}, r_{1212}$ are as in \eqref{eq:paramrepl}). The eigenvalues of $C(\Psi)$ are
\begin{IEEEeqnarray*}{CL}
 -r, &\\
 r &  \text{with multiplicity } 3,\\
 0 & \text{with multiplicity } 3, \\ 
   \frac{1}{2}( r \!+\! r_0 \!+\! r_1 \!+\! r_2) \pm & 
   \sqrt{3 c^2 \!+\! r_0^2 \!+\! r_1^2 \!+\! r_2^2
\!-\! r_0 r_1  \!-\!  r_0 r_2 \!-\! r_1 r_2 },
\end{IEEEeqnarray*}
which shows that $C(\Psi)$ can be positive semidefinite only if  $r=0$.
In the special case, when $r=0$, the matrix 
$ \Psi \left( \ketbra{\x}{\x} \right)$ is
\begin{equation*}
    \begin{bmatrix}
r_0\, x_0 \ols{x_0}  & 
-\frac{1}{2} (r_0 \!+\! r_1 \!-\! r_2  \!-\! 2 i c)  x_0 \ols{x_1} & 
- \frac{1}{2} (  r_0 \!-\! r_1 \!+\! r_2 \!+\! 2 i c)  x_0 \ols{x_2} \\
 -\frac{1}{2} (r_0 \!+\! r_1 \!-\! r_2  \!+\! 2 i c)  x_1 \ols{x_0} & 
 r_1 x_1 \ols{x_1}  & 
 \frac{1}{2} (r_0 \!-\! r_1 \!-\! r_2  \!+\! 2 i c)  x_1 \ols{x_2} \\ 
 -\frac{1}{2} (r_0 \!-\! r_1 \!+\! r_2  \!-\! 2 i c)   x_2 \ols{x_0} & 
\frac{1}{2} (r_0 \!-\! r_1 \!-\! r_2  \!-\! 2 i c)   x_2 \ols{x_1} &
r_2  x_2 \ols{x_2}
    \end{bmatrix},
\end{equation*}
whose determinant is zero and its $2 \times 2$ principal minors are equal to 
\begin{equation*}
   \left( 
   r_0 r_1 + r_0 r_2 + r_1 r_2- \frac{1}{4} (r_0 + r_1 + r_2)^2- c^2
   \right) x_i \ols{x_i}\,  x_j \ols{x_j} \ \ \text{ for } 
     i \neq j. 
\end{equation*}
This proves that for $r=0$ the map $\Psi$ is completely positive under  the following conditions: $r_0 \geq 0,\, r_1 \geq 0,\, r_2 \geq 0$
and
$2 (r_0 r_1 + r_0 r_2 + r_1 r_2) - r_0^2 - r_1^2 - r_2^2 -4 c^2
\geq 0$.
(Geometrically, these conditions define a conic in $\RR^3_{+}$, which we denote by $\mathcal{C}_c$; specifically $\mathcal{C}_0$ is a cone with vertex $(0,0,0)$ and axis along $(1,1,1)$.)  
On the other hand, when $r \neq 0$, we will from now on 
assume that $r=1$. It remains to determine the conditions when $\Psi$ is positive. 
The $1 \times 1$ and $2 \times 2$  principal minors \eqref{eq:minors1} and \eqref{eq:minors2} of $\Psi(\ketbra{\x}{\x})$ are everywhere nonnegative if 
\begin{IEEEeqnarray}{C} \label{eq:ChoiCondit}
r_0 \geq 0,\ r_1 \geq 0,\ r_2 \geq 0, \nonumber  \\
   4 r_0 r_1 - (1 + r_0 + r_1 - r_2)^2 - 4 c^2 \geq 0, \\
4 r_0 r_2 - (1 + r_0 - r_1 + r_2)^2 - 4 c^2  \geq 0, \nonumber \\
 4 r_1 r_2 - (1 - r_0 + r_1 + r_2)^2 - 4 c^2 \geq 0. \nonumber
\end{IEEEeqnarray} 
We claim that these conditions also imply $\det \Psi(\ketbra{\x}{\x})\geq 0$. Indeed, if we write 
$R:=r_0 r_1 + r_0 r_2 + r_1 r_2- 1/4 (1 + r_0 + r_1 + r_2)^2- c^2$, then the inequalities \eqref{eq:ChoiCondit} become 
$r_0 \geq 0,\ r_1 \geq 0,\ r_2 \geq 0$ and
$R + r_2 \geq 0, \ R + r_1 \geq 0,\ R + r_0 \geq 0$,
and the determinant \eqref{eq:minors3} is
\begin{IEEEeqnarray*}{C} 
(R+r_2) s_0^2 s_1 +  
 (R+r_1) s_0 s_2^2 + 
(R+r_0) s_1^2 s_2 + \\
   r_1\, s_0 s_1^2 +    r_0\, s_0^2 s_2  + 
  r_2\, s_1 s_2^2 - \\
   \left( 3 R + 2 r_0 + 2 r_1 + 2 r_2 
    \right) s_0 s_1 s_2,
\end{IEEEeqnarray*}
which is nonnegative for all $s_0=x_0 \ols{x_0},\, s_1=x_1 \ols{x_1},\, s_2=x_2 \ols{x_2}$ by 
the weighted AM-GM inequality (the inequality of arithmetic and geometric means). We combine the two constructions above and obtain the complete 5 parameter family of positive maps, 
\begin{IEEEeqnarray*}{C}
\Psi_{r,r_0,r_1,r_2,c} \colon \ \ 
 \left[ z_{ij} \right]_{i,j=0}^2  \ \ \mapsto \vspace{1mm}\\
    \begin{bmatrix}
r_0\, z_{00} +  r\, z_{11}  & 
( \substack{-\frac{1}{2} (r + r_0 + r_1 - r_2)} + i\, c) z_{01} & 
(\substack{-\frac{1}{2} (r + r_0 - r_1 + r_2)} - i\, c) z_{02} \\
 (\substack{-\frac{1}{2} (r + r_0 + r_1 - r_2)} - i\, c) z_{10} & 
 r_1\, z_{11} + r\, z_{22}  & 
(\substack{-\frac{1}{2} (r - r_0 + r_1 + r_2)}+ i\, c ) z_{12} \\ 
(\substack{-\frac{1}{2} (r + r_0 - r_1 + r_2)} + i\, c)  z_{20} & 
(\substack{-\frac{1}{2} (r - r_0 + r_1 + r_2)} - i\, c )  z_{21} &
 r\,  z_{00} +r_2\, z_{22}
    \end{bmatrix},
\end{IEEEeqnarray*}
where $r, r_0,r_1,r_2,c$ satisfy the inequalities 
\begin{IEEEeqnarray*}{C} 
r \geq 0,\ r_0 \geq 0,\ r_1 \geq 0,\ r_2 \geq 0,   \\
R + r r_2 \geq 0, \ R + r r_1 \geq 0,\ R + r r_0 \geq 0
\end{IEEEeqnarray*} 
and $R$ stands for $r_0 r_1 + r_0 r_2 + r_1 r_2- 1/4 (r + r_0 + r_1 + r_2)^2- c^2$.

Finally, we select the maps that lie on the extreme rays of the positive cone. Recall that the above construction returns all positive maps $\Psi$ such that the associated form $p_{\Psi}$ is zero on the set 
$\mathcal{Z}_C$ prescribed in \eqref{eq:zerosChoiC}. Observe that 
$\Psi_{r,r_0,r_1,r_2,c} $ is positive if and only if $\Psi_{r,r_0,r_1,r_2,-c}$ is positive.
When positive map $\Psi_{r,r_0,r_1,r_2,c}$ has $c \neq 0$, then 
for any $\lambda \in (0,1)$ we get
$\lambda \Psi_{r,r_0,r_1,r_2,c} + (1-\lambda) \Psi_{r,r_0,r_1,r_2,-c}=
 \Psi_{r,r_0,r_1,r_2,(2 \lambda -1)c}$  another positive map which which is not extremal.
Next we present a geometric argument showing that
any  $\Psi_{1,r_0,r_1,r_2,c}$ can be written as the sum of nonzero positive and completely positive maps, unless at least one of the conditions $r_0=r_1,\, r_0=r_2,\, r_1=r_2$ is true (in particular this implies the extremality of the Choi map \eqref{eq:ChoisMap} with $r_0=r_1=r_2=1$ and $c=0$). 
Positive maps $\Psi_{1,r_0,r_1,r_2,c}$ are determined by
the inequalities \eqref{eq:ChoiCondit},
which parametrise a convex set with vertex 
$(\frac{1}{3}(1+2\sqrt{1+3 c^2}) \times (1,1,1)$ obtained  as the intersection of three conics, namely the translations of 
$\mathcal{C}_c$ by 
$(1,0,0), (0,1,0)$ and $(0,0,1)$. Specifically, when $c=0$, positive maps $\Psi_{1,r_0,r_1,r_2,0}$ are determined by the intersection of cones
$(1,0,0)+\mathcal{C}_0$,
$(0,1,0)+\mathcal{C}_0$ and $(0,0,1)+\mathcal{C}_0$.
Given an extremal positive map $\Psi_{1,r_0,r_1,r_2,c}$, we can
without loss of generality assume  $1\leq r_0\leq r_1\leq r_2$ and that $(r_0,r_1,r_2)$ lies on the boundary of the conic $(1,0,0)+\mathcal{C}_c$, i.e.,   $r_0 r_1 + r_0 r_2 + r_1 r_2- 1/4 (1 + r_0 + r_1 + r_2)^2- c^2+r_0=0$. Then
$(r_0-1,r_1,r_2)$ lies on $\mathcal{C}_c$, which implies that $\Psi_{0,r_0-1,r_1,r_2,c}$ is completely positive. If furthermore 
$r_0<r_1$, there exists an $\varepsilon>0$ such that  
$(1-\varepsilon)r_0 +\varepsilon \leq (1-\varepsilon) r_1 \leq 
(1-\varepsilon) r_2$. Then it is straightforward to verify that 
$((1-\varepsilon)r_0 +\varepsilon , (1-\varepsilon) r_1 , (1-\varepsilon) r_2) \in (1,0,0)+\mathcal{C}_{(1-\varepsilon) c}$,
and consequently  $\Psi_{1,(1-\varepsilon)r_0 +\varepsilon , (1-\varepsilon) r_1 , (1-\varepsilon) r_2,(1-\varepsilon) c}$ is a positive map.
 Since 
 \begin{IEEEeqnarray*}{CCC} 
\Psi_{1,r_0,r_1,r_2,c} & = & \Psi_{0,\varepsilon (r_0-1), \varepsilon r_1, \varepsilon r_2, \varepsilon c} + \Psi_{1,(1-\varepsilon)r_0 +\varepsilon , (1-\varepsilon) r_1 , (1-\varepsilon) r_2,(1-\varepsilon) c} \\
& = & \varepsilon \, \Psi_{0,r_0-1, r_1,  r_2,  c} + \Psi_{1,(1-\varepsilon)r_0 +\varepsilon , (1-\varepsilon) r_1 , (1-\varepsilon) r_2,(1-\varepsilon) c},
\end{IEEEeqnarray*} 
we obtain a contradiction with the extremality, unless it holds $r_0=r_1$.
\end{exmp}

\section{New examples of entanglement witnesses}
\label{sec:examples}

In~\cite{BucSiv} we constructed new families of 
nonnegative forms $p_{\Psi}$ on $\PP^2(\CC) \times \PP^2(\CC)$ which
have 8, 9 or 10 real zeros. 
This ensures that the associated positive maps
$\Psi \colon \M_3^{\text{sa}} \rightarrow \M_3^{\text{sa}}$
are significantly different from the Choi map \eqref{eq:ChoisMap} 
and its generalizations considered in \cref{ex:Choi}, whose corresponding forms have 7 real zeros.

From the examples in~\cite{BucSiv} we want to choose positive maps that are \textit{new} entanglement witnesses, i.e.,
they detect entanglement in some states which the PPT criterion fails to certify. In order to achieve this, we can only pick the quantum maps which are not decomposable. Indeed, a positive map  $\Psi \in  \boldsymbol{DEC}$ is by definition of the form
$\Psi= \Psi_{\text{cp}_1} + T \circ \Psi_{\text{cp}_2}$ for some 
completely positive maps $\Psi_{\text{cp}_1}$ and $\Psi_{\text{cp}_2}$.
Then, by the correspondences in \cref{tab:conesO,tab:conesSoO}, it holds
\begin{equation*}
    C(\Psi) \in \text{co-}\!\PSD + \PSD = \mathcal{PPT}^{\ast}:=
    \{\tau \colon \langle \tau, \rho
    \rangle_{\text{HS}} \geq 0, \forall \rho \in \mathcal{PPT}  \}. 
\end{equation*}
This shows that whenever a decomposable $\Psi$ witnesses entanglement in a state $\rho$, the entanglemnt of $\rho$ is already detected by the transposition (i.e., by the PPT criterion). 

Consequently, a positive map $\Psi$ which lies on an extreme ray of the positive cone 
is decomposable if and only if  either $\Psi \in \CP$
or  $T \circ \Psi \in \CP$. This implies the following advantage of optimal entanglement witnesses.
An extremal positive map that is neither CP nor co-CP defines, not only an optimal entanglement witness, but also a new entanglement witness that confirms entanglement in some new PPT states.  

\subsection{The Buckley-\v Sivic maps} 
\label{subsec:BucSiv}
From the prescribed sets of zeros we can simultaneously determine the positivity, complete positivity and extremality for the entire families of quantum maps (see \cref{subsec:construction}). We consider 
examples of positive maps on $\M_3^{\text{sa}}$ from \cite{BucSiv}, 
which  belong to extreme rays in the cone of positive maps.
Moreover, we will select quantum maps that are not decomposable, therefore the linear functionals 
$\langle C(\Psi), \cdot \, \rangle_{\text{HS}}=
\Tr \left( C(\Psi) \, \cdot \, \right)$
are new optimal entanglement witnesses.

\begin{exmp}[10 real zeros] \label{ex:BS10}
We present the family of positive maps corresponding to 
$p_{\Psi}(\x, \y)$ with the following set of zeros in $\CC$
\begin{multline*} \mathcal{Z}_t :=
 \left\{     (e^{i \varphi_0},e^{i \varphi_1},e^{i \varphi_2};\, 
      e^{i \varphi_0},e^{i \varphi_1},e^{i \varphi_2}), \right.  \\
\left. (1,t e^{i \varphi},0;\, t e^{-i \varphi},1,0), 
(0,1, t e^{i \varphi};\, 0, t e^{-i \varphi}, 1), 
(t e^{i \varphi}, 0,1;\, 1,0,t e^{-i \varphi})
 \right\},
\end{multline*}
where $\varphi, \varphi_0,\varphi_1,\varphi_2 \in [0, 2 \pi)$ and 
$t \in \RR$; only 10 of them have real coordinates,
\begin{multline*}
      (1,1,1; 1,1,1), (1,1,-1;1,1,-1), (1,-1,1;1,-1,1), (-1,1,1;-1,1,1), \\
(1,t,0;t,1,0), (0,1,t,0,t,1), (t,0,1;1,0,t),\\ 
 (1,-t,0; -t,1,0), (0,1,-t;0,-t,1), (-t,0,1;1,0,-t).
\end{multline*}

\begin{theorem} \label{thm:BS10} Quantum maps 
$\Psi_t \colon \M_3^{\text{sa}} \to \M_3^{\text{sa}}$ 
of the form 
\begin{equation*}
 \begin{array}{c}
\begin{bmatrix}
    z_{00}& z_{01}&z_{02} \\
z_{10}&   z_{11} &z_{12}\\
z_{20} &z_{21} &  z_{22}
\end{bmatrix}
\\
\begin{tikzcd}
  \arrow[mapsto]{d}
  \\
\phantom{a} 
\end{tikzcd}\\
\begin{bmatrix}
(t^2-1)^2 z_{00}+z_{11}+t^4 z_{22} \!\!&\!\! -(t^4-t^2+1)\, z_{01}
\!\!&\!\! -(t^4-t^2+1)\, z_{02} \\
-(t^4-t^2+1)\, z_{10} \!\!&\!\!  t^4 z_{00}+(t^2-1)^2 z_{11}+z_{22} \!\!&\!\!
-(t^4-t^2+1)\, z_{12}\\
-(t^4-t^2+1)\, z_{20} \!\!&\!\! -(t^4-t^2+1)\, z_{21}  \!\!&\!\!
z_{00}+t^4 z_{11}+ (t^2-1)^2 z_{22}
\end{bmatrix}\\
\end{array}   
\end{equation*}
are  positive for all $ t\in \mathbb{R}.$
Apart from $t=\pm 1$, every $\Psi_t$ defines an extreme ray in the cone of positive maps and is neither completely nor co-completely positive. 
\end{theorem}
\begin{proof} 
The above maps are the solutions of the linear system 
of equations \eqref{eq:SAcoeff} and \eqref{eq:zeroCoeff} for every $(\x, \y) \in \mathcal{Z}_t$. 
First we verify that $\Psi_t$ are positive maps. The matrix
$\Psi_t\left(  \ketbra{\x}{\x} \right)$ is positive semidefinite for all $\x =(x_0, x_1,x_2) \in \CC^3$ if
the following conditions hold:
\begin{enumerate}
    \item $\Tr \Psi_t\left(  \ketbra{\x}{\x} \right) = 
    2 \left(1 - t^2 + t^4 \right) \left(|x_0|^2  + |x_1|^2 + |x_2|^2 \right)
    \geq 0$,
    \item the sum of the principal $2 \times 2$ minors 
    $\left( 1 - t^2 + t^4 \right)^2 
    \left( |x_0|^2  + |x_1|^2 + |x_2|^2 \right)^2  \geq 0$,
     \item $\det \Psi_t\left(  \ketbra{\x}{\x} \right) \geq 0$, 
     where \vspace{-5mm}
\end{enumerate}
\begin{multline*} 
    \det \Psi_t\left(  \ketbra{\x}{\x} \right) \ \ \  = \ \ \ 
    (1-t^2)^2\times  \vspace{1mm} \\
     \left[   t^4 \left(|x_0|^6   \!+\! |x_1|^6  \!+\! |x_2|^6 \right) +
    (t^8 \!-\! 2 t^2) \left(|x_0|^4 |x_1|^2  \!+\! 
   |x_0|^2 |x_2|^4    \!+\! |x_1|^4  |x_2|^2 \right)+ \right.  \vspace{1mm} \\
   \left.  (1 \!\!-\!\! 2 t^6) \left(|x_0|^2  |x_1|^4  \!+\! 
   |x_0|^4  |x_2|^2  \!+\!   |x_1|^2 |x_2|^4 \right)
   \!-\! 3 (1 \!\!-\!\! 2 t^2 \!\!+\!\! t^4 \!\!-\!\! 2 t^6 \!\!+\!\! t^8) |x_0|^2|x_1|^2|x_2|^2  \right].
\end{multline*}
Polynomials in 1. and 2. are nonnegative,
and 3. equals the \textit{generalized Robinson polynomial}~\cite{Rezn}, which is known to be positive everywhere except in 10 zeros.

Next we calculate the Choi matrix $C(\Psi_t)$, 
\begin{equation*}
 \left[   \begin{array}{ccc|ccc|ccc}
\!\!\! \substack{\left(t^2-1 \right)^2 } \!\!&\!\! \cdot &  \cdot  &  
\cdot   \!\!&\!\!
  \substack{-1 + t^2 - t^4 } \!\!&\!\! \cdot  & 
\cdot  &  \cdot   \!\!&\!\!
  \substack{-1 + t^2 - t^4 } \!\!\!  \\
 \cdot  & 1 &  \cdot & \cdot &  \cdot  &  \cdot  &
   \cdot  &  \cdot  &  \cdot  \\
 \cdot  &  \cdot  & t^4 &
  \cdot  &  \cdot  &  \cdot & \cdot & \cdot  &  \cdot  \\
\hline
 \cdot  &\cdot &  \cdot  &
 t^4 &  \cdot  &  \cdot  &
  \cdot  &  \cdot  &  \cdot  \\
\!\!\! \substack{-1 + t^2 - t^4 }  \!\!&\!\!   \cdot  &  \cdot  &
 \cdot \!\!&\!\! \substack{\left(t^2-1 \right)^2} \!\!&\!\!  \cdot &
  \cdot  &  \cdot   \!\!&\!\!  \substack{-1 + t^2 - t^4 }  \!\!\!  \\
 \cdot  &  \cdot &  \cdot  &
  \cdot  &  \cdot  & 1 &
   \cdot & \cdot & \cdot \\
\hline
 \cdot  &  \cdot  & \cdot &
  \cdot  &  \cdot  &  \cdot &
  1  &  \cdot  &  \cdot \\
 \cdot  &  \cdot  &  \cdot  &
  \cdot  &  \cdot  & \cdot & 
   \cdot & t^4 &  \cdot \\
\!\!\!  \substack{-1 + t^2 - t^4 } \!\!&\!\!  \cdot  &  \cdot  &
 \cdot  \!\!&\!\!  \substack{-1 + t^2 - t^4 } \!\!&\!\!  \cdot  &
  \cdot &  \cdot  \!\!&\!\! \substack{\left(t^2-1 \right)^2} \!\!\!
\end{array} \right].
\end{equation*}
The eigenvalues of $C(\Psi_t)$ are 
\begin{equation*}
    1,\, 1,\, 1,\, t^4,\, t^4,\, t^4,\, -(1 + t^4),\, 2 - 3 t^2 + 2 t^4,\, 2 - 3 t^2 + 2 t^4,
\end{equation*}
therefore $\Psi_t$ is not completely positive since $ -(1 + t^4)$ is negative for all $t \in \RR$. For an alternative proof see \cref{rem:Rzeros} and the discussion on the number of real zeros.

Finally, $\Psi_t$ is an extremal map since it is,
up to a positive factor, the only positive linear map for which $p_{\Psi}$ has the prescribed set of zeros.
It remains to be checked whether there exist $t$ for which $\Psi_t$ is co-completely positive, i.e., $T \circ \Psi_t$ is completely positive. 
To this end we  calculate the Choi matrix $C(T \circ  \Psi_t)$ 
\begin{equation*}
 \left[   \begin{array}{ccc|ccc|ccc}
\!\!\! \substack{\left(t^2-1 \right)^2 } \!\!&\!\! \cdot &  \cdot  &  
\cdot  & \cdot & \cdot  & 
\cdot  &  \cdot   & \cdot \\
 \cdot  & 1 &  \cdot  \!\!&\!\!
  \substack{-1 + t^2 - t^4 } \!\!&\!\!  \cdot  &  \cdot  &
   \cdot  &  \cdot  &  \cdot  \\
 \cdot  &  \cdot  & t^4 &
  \cdot  &  \cdot  &  \cdot  \!\!&\!\!
  \substack{-1 + t^2 - t^4 } \!\!&\!\! \cdot  &  \cdot  \\
\hline
 \cdot   \!\!&\!\! \substack{-1 + t^2 - t^4 } \!\!&\!\!  \cdot &
 t^4 &  \cdot  &  \cdot  &
  \cdot  &  \cdot  &  \cdot  \\
 \cdot &  \cdot  &  \cdot  &
 \cdot \!\!&\!\! \substack{\left(t^2-1 \right)^2} \!\!&\!\! \cdot &
  \cdot  &  \cdot  & \cdot \\
 \cdot  &  \cdot &  \cdot  &
  \cdot  &  \cdot  & 1 &
   \cdot  \!\!&\!\! \substack{-1 + t^2 - t^4 } \!\!&\!\! \cdot \\
 \hline
 \cdot  &  \cdot  \!\!&\!\! \substack{-1 + t^2 - t^4 } \!\!&\!\!
  \cdot  &  \cdot  &  \cdot &
  1  &  \cdot  &  \cdot \\
 \cdot  &  \cdot  &  \cdot  & 
 \cdot  & \cdot  \!\!&\!\!  \substack{-1 + t^2 - t^4 } \!\!&\!\!
   \cdot & t^4 &  \cdot \\
\cdot & \cdot  &  \cdot  &
 \cdot  & \cdot &  \cdot  &
  \cdot &  \cdot  \!\!&\!\! \substack{\left(t^2-1 \right)^2} \!\!\!
\end{array} \right],
\end{equation*}
whose eigenvalues are
\begin{equation*}
    (-1 + t^2)^2,\   \frac{1}{2} \left(1 + t^4 \pm
    \sqrt{ 5 - 8 t^2 + 10 t^4 - 8 t^6 + 5 t^8} \right),
\end{equation*}
each with multiplicity 3. It is easy to verify that 
$1 + t^4 - \sqrt{ 5 - 8 t^2 + 10 t^4 - 8 t^6 + 5 t^8}$ is negative except for $t= \pm1$ when it attains value zero.
This proves that $T \circ \Psi_{t}$ is not completely positive unless $t= \pm1$. The completely positive map $T \circ \Psi_{\pm 1}$
has the Kraus decomposition 
\begin{equation*}
  T \circ  \Psi_{\pm 1}\left(  Z \right) = 
    A_1 Z A_1^{\ast}  + A_2 Z A_2^{\ast} 
   +A_3 Z A_3^{\ast},
\end{equation*}
where
\begin{equation*}
    A_1= \begin{bmatrix}
    0 & 1 & 0 \\ 
    -1 & 0 & 0 \\
    0& 0& 0
    \end{bmatrix}, \ 
    A_2= \begin{bmatrix}
    0 & 0 & -1 \\ 
    0 & 0 & 0 \\
    1 & 0 & 0
    \end{bmatrix}, \
    A_3 = \begin{bmatrix}
    0 & 0 & 0 \\ 
    0 & 0 & 1 \\
    0 & -1 & 0
    \end{bmatrix}.
\end{equation*}
\end{proof}

\begin{remark}
For $t=0$, the set of zeros $\mathcal{Z}_t$ is the Choi set of zeros
\eqref{eq:zerosChoiC} considered in \cref{ex:Choi}. The map $\Psi_0$ 
is equal to the Choi map \eqref{eq:ChoisMap}.
The extreme positive maps in \cref{thm:BS10} were first found in 
\cite{HaKyeE}, as an exposed and indecomposable subset of the Choi-type maps in \cite{HaKye}, by using alternative methods to ours. 
Note also that the families of positive maps in \cite{HaKye} and \cref{ex:Choi} overlap. 
\end{remark}

\end{exmp}

\begin{exmp}[9 real zeros] \label{ex:BS9}
Next we present the family of positive maps corresponding to 
$p_{\Psi}(\x, \y)$ with the following zeros,
\begin{multline*} \label{eq:zerospq}
\mathcal{Z}_{p,q} := \left\{
      (e^{i \varphi_0},e^{i \varphi_1},e^{i \varphi_2};\, 
      e^{i \varphi_0},e^{i \varphi_1},e^{i \varphi_2}), \right. \\
\left. (1,p\, e^{i \varphi},0;\, q\, e^{-i \varphi},1,0), 
 (0,1,q\, e^{i \varphi};\, 0,p\, e^{-i \varphi},1), (0,0,1;\, 1,0,0) 
 \right\},
\end{multline*}
where $p,q\in\RR$ and $\varphi, \varphi_0,\varphi_1,\varphi_2 \in [0, 2 \pi)$.
Note that 9 of these zeros are real,
\begin{multline*}
      (1,1,1; 1,1,1), (1,1,-1;1,1,-1), (1,-1,1;1,-1,1), (-1,1,1;-1,1,1), \\
(1,p,0; q,1,0), (1,-p,0; -q,1,0), 
 (0,1,q; 0,p,1), (0,1,-q; 0,-p,1),  (0,0,1; 1,0,0).
\end{multline*}

\begin{theorem} \label{thm:BS9}
Superoperators $\Psi_{p,q} \colon \M_3^{\text{sa}} \to \M_3^{\text{sa}}$ 
of the form 
\begin{equation*}
 \begin{array}{c}
Z
\\
\begin{tikzcd}
  \arrow[mapsto]{d}
  \\
\phantom{a} 
\end{tikzcd}\\
\begin{bmatrix} \!\!
\substack{p^2\!(p q\!-\!1)^2\! z_{00}\!+\!q(2p\!-\!q)z_{11}}  
\!\!\!\!\!\!&\!\!\!\!\!\! 
\substack{-p q (1\!-\!q^2\!+\!p^2\! q^2) z_{01}} 
\!\!\!\!\!\!\!\!\!\!\!\!&\!\!\!\!\!\!\!\!\!\!\!\!
\substack{(p q\!-\!1)(p^2\!+\!p q\!-\!p^3 q\!-\!q^2\!+\!p^2 q^2) z_{02}} \\
\substack{-p q (1\!-\!q^2\!+\!p^2\! q^2) z_{10}} 
\!\!\!\!\!\!&\!\!\!\!\!\!  
\substack{p^2\! q^3\! (2p\!-\!q) z_{00} \!+\! q^2\!(p q\!-\!1)^2\! z_{11}
\!+\!q (2p\!-\!q)z_{22} } 
\!\!\!\!\!\!\!\!\!\!\!\!&\!\!\!\!\!\!\!\!\!\!\!\! 
\substack{-p q(1\!-\!q^2\!+\!p^2\! q^2) z_{12} }\\
\!\! \substack{(p q\!-\!1)(p^2\!+\!p q\!-\!p^3\! q\!-\!q^2\!+\!p^2\! q^2) z_{20}} \!\!\!\!\!\!&\!\!\!\!\!\!
\substack{-p q(1\!-\!q^2\!+\!p^2\! q^2) z_{21} } 
\!\!\!\!\!\!\!\!\!\!\!&\!\!\!\!\!\!\!\!\!\!\!\! 
\substack{q(2p\!-\!q)(1\!-\!p^2\!q^2)z_{00} \!+\! 
p^2\! q^3\!(2p\!-\!q) z_{11} \!+\! p^2\!(p q\!-\!1)^2\! z_{22} }
\!\! \end{bmatrix}\\
\end{array}   
\end{equation*}
are  extreme positive maps in the region
\begin{displaymath}  \mathcal{R} =
\set{ (p,q) \in (0,\frac{1}{\sqrt{2}}) \times (0,\sqrt{2}) \colon
2p-q \geq 0, (p^2-1)^2q^2-p^2 \geq 0}.
\end{displaymath}
Apart from the segment $2p=q$, they are neither completely nor co-completely positive.
\end{theorem}

\begin{figure}[htbp]
  \centering
  \includegraphics[width=0.3\textwidth]{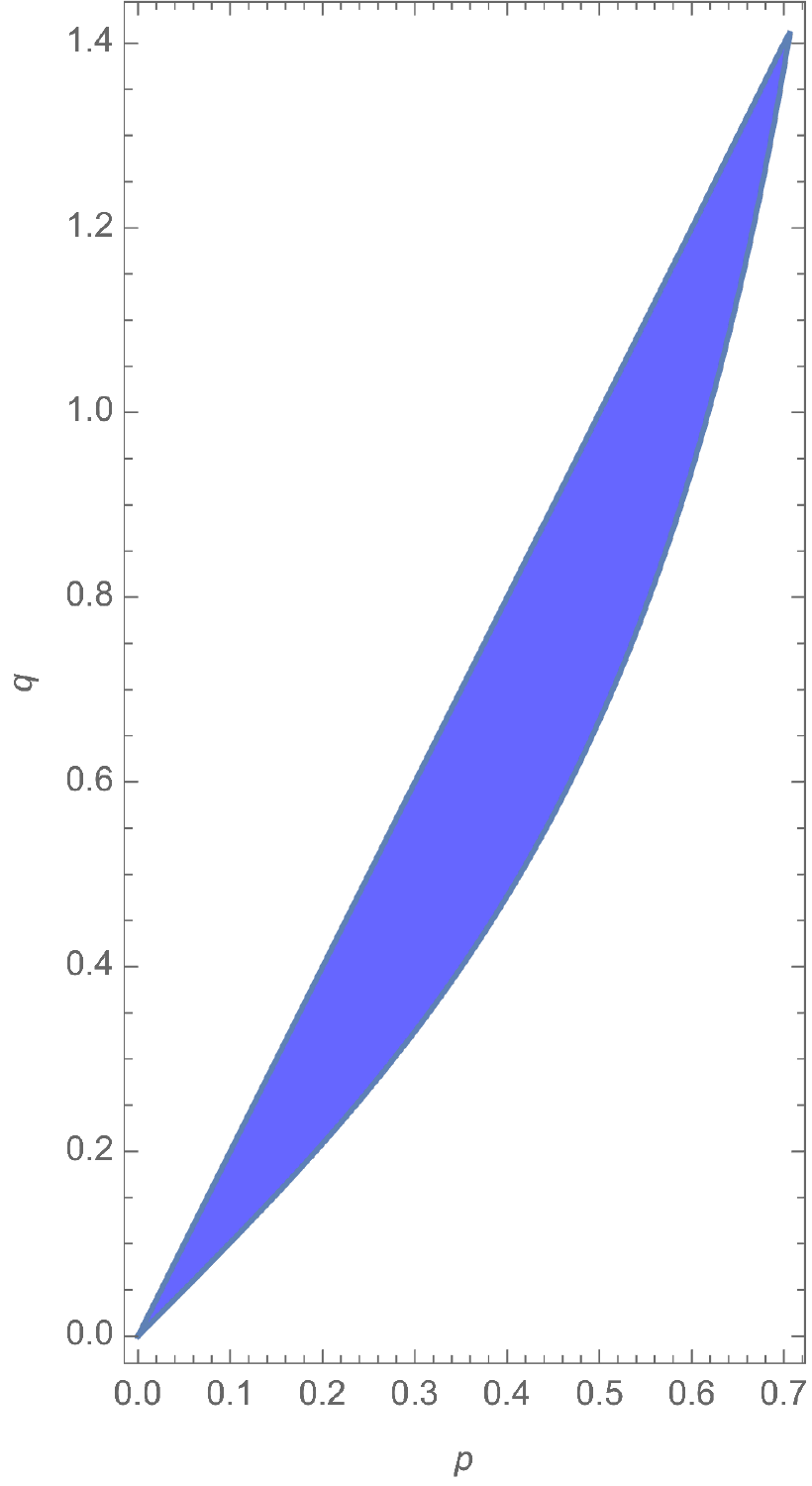}
   \hspace{8mm}
   \includegraphics[width=0.3\textwidth]{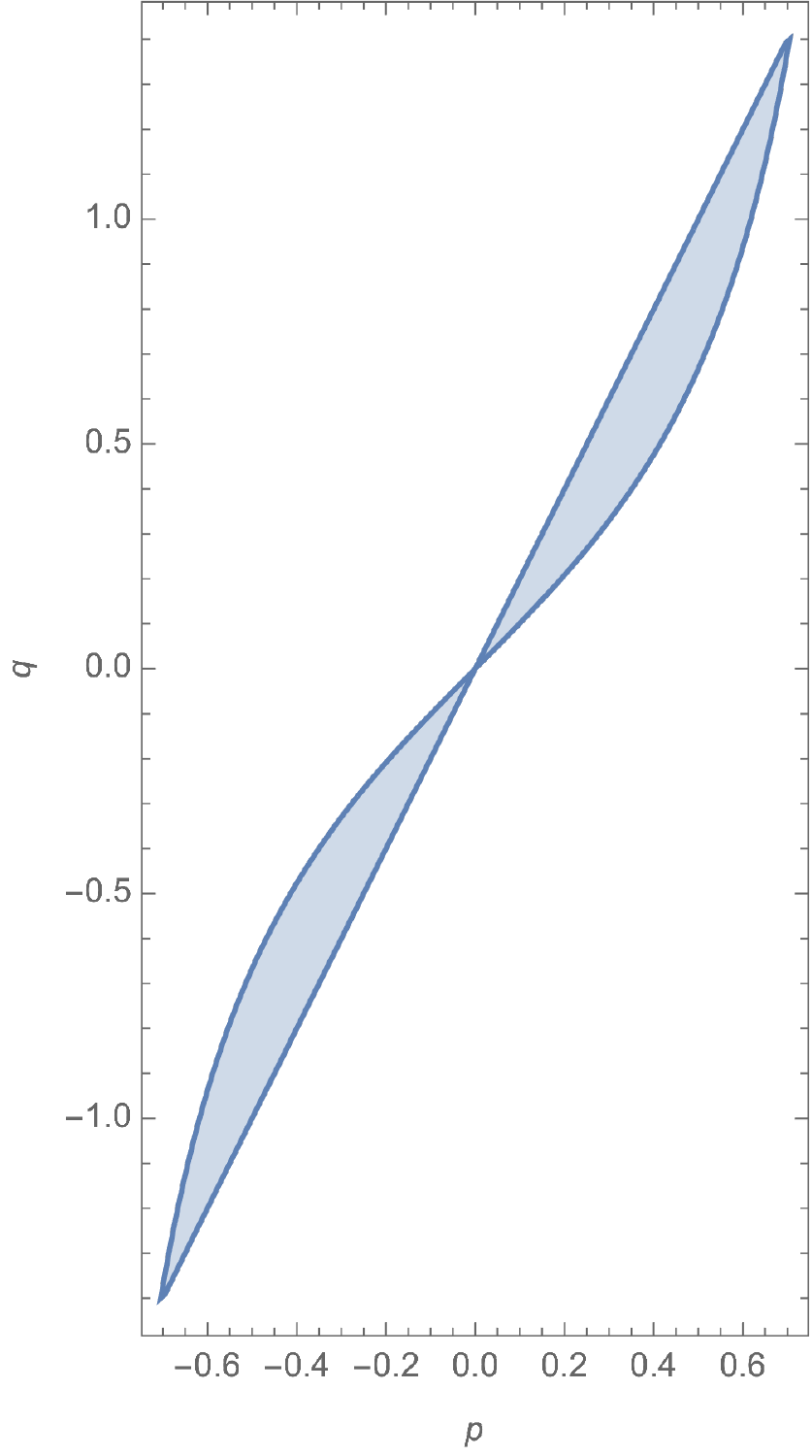}
   \put(-210,130){$\substack{2p-q \geq 0}$}
   \put(-180,75){$\substack{q \geq \frac{p}{1-p^2}}$}
  \caption{Region of points $(p,q) \in  \mathcal{R}$ in \cref{thm:BS9}
  for which $\Psi_{p,q}$ are positive maps. }
  \label{fig:leafpq}
\end{figure}

\begin{proof}
See \cite{BucSiv} for the complete proof that $p_{\Psi_{p,q}}$ 
with the prescribed zeros $\mathcal{Z}_{p,q}$ are 
nonnegative and extremal in the region $\mathcal{R}$ in \cref{fig:leafpq}.
As remarked in \cref{subsec:construction}, the main part constitutes of the proofs of positivity, whereas extremality follows directly from prescribing large enough sets of zeros. 

The Choi matrix $C(\Psi_{p,q})$ is 
\begin{equation*}
 \left[   \begin{array}{ccc|ccc|ccc}
\substack{p^2} c_{00} & \cdot &  \cdot  &  
\cdot  & c_{04} & \cdot  & 
\cdot  &  \cdot   & c_{08}   \\
 \cdot  & c_{11} &  \cdot & \cdot &  \cdot  &  \cdot  &
   \cdot  &  \cdot  &  \cdot  \\
 \cdot  &  \cdot  & 0 &
  \cdot  &  \cdot  &  \cdot & \cdot & \cdot  &  \cdot  \\ 
\hline 
 \cdot  &\cdot &  \cdot  &
 \substack{p^2 q^2} c_{11} &  \cdot  &  \cdot  &
  \cdot  &  \cdot  &  \cdot  \\
c_{04} & \cdot  &  \cdot  &
 \cdot & \substack{q^2} c_{00} & \cdot &
  \cdot  &  \cdot  & c_{04}  \\
 \cdot  &  \cdot &  \cdot  &
  \cdot  &  \cdot  & c_{11} &
   \cdot & \cdot & \cdot \\
\hline
 \cdot  &  \cdot  & \cdot &
  \cdot  &  \cdot  &  \cdot &
  \substack{(1 - p^2 q^2)} c_{11}  &  \cdot  &  \cdot \\
 \cdot  &  \cdot  &  \cdot  &
  \cdot  &  \cdot  & \cdot & 
   \cdot & \substack{p^2 q^2} c_{11} &  \cdot \\
c_{08} & \cdot  &  \cdot  &
 \cdot  & c_{04} &  \cdot  &
  \cdot &  \cdot  & \substack{p^2} c_{00}
\end{array} \right],
\end{equation*}
where 
\begin{align*}
    c_{00}   := & (1 - p q)^2  & c_{04}   := & -p q 
    \left(1 + (p^2-1) q^2 \right)  \\
     c_{11}   := & (2 p - q) q & c_{08}   := & -(1 - p q) 
     \left( p^2  + (p^2-1)  (q^2-p q)  \right),
\end{align*}
and  its eigenvalues are 
\begin{multline*}
    0,\, 2 \times c_{11},\, (1-p^2 q^2) c_{11},\, 
    p^2 c_{00}-c_{08},\,
    2 \times p^2 q^2 c_{11},  \\
    \frac{1}{2} \left(c_{08} + (p^2 +q^2) c_{00} \pm
   \sqrt{(c_{08} + (p^2 +q^2) c_{00} )^2 + 
    4 (2 c_{04}^2 - q^2 c_{00} c_{08}  - p^2 q^2 c_{00}^2 )} \right).
\end{multline*}
A simple plot in \texttt{Mathematica} verifies that 
in the region $\mathcal{R}$
only the last of these eigenvalues attains negative value. 
In particular,
for $(p,q) \in \mathcal{R}$ it holds 
\begin{IEEEeqnarray*}{C}
c_{08} + (p^2 +q^2) c_{00} -
   \sqrt{(c_{08} + (p^2 +q^2) c_{00} )^2 + 
    4 (2 c_{04}^2 - q^2 c_{00} c_{08}  - p^2 q^2 c_{00}^2 )} < 0 \\
    \Longleftrightarrow \\
    2 c_{04}^2 - q^2 c_{00} c_{08}  - p^2 q^2 c_{00}^2 > 0\\
     \Longleftrightarrow \\
    (2 p - q) q^2 \Big(
    p (1+p^2 q^2) \left( (1-p q)^2 +(2 p -q)q \right)+ q(1-p q)^2 
    \Big) > 0. 
\end{IEEEeqnarray*}
From the last inequality we conclude that,
 for $(p,q) \in \mathcal{R}$, the Choi matrix
$C(\Psi_{p,q})$ is positive semidefinite precisely on the line segment $2p-q = 0$.
The completely positive map $\Psi_{p,2p}$ has Kraus rank 1 since we can write,
$\Psi_{p,2p}(Z)=$
\begin{equation*} \label{pg:pqCP}
    p^2 (2 p^2-1)^2
    \begin{bmatrix}
        z_{00} & -2 z_{01} & z_{02}\\
        -2 z_{10} & 4 z_{11} & -2 z_{12} \\
        z_{20} & -2 z_{21} & z_{22}
    \end{bmatrix}= 
    p^2 (2 p^2-1)^2
     \begin{bmatrix}
        1 & 0 & 0 \\
        0 & -2 & 0 \\
        0 & 0 & 1
    \end{bmatrix}
      Z 
    \begin{bmatrix}
        1 & 0 & 0 \\
        0 & -2 & 0 \\
        0 & 0 & 1
    \end{bmatrix}.
\end{equation*}

Next we verify whether any $\Psi_{p,q}$ is co-completely positive. To this end we compute the Choi matrix $C\left( T \circ \Psi_{p,q} \right)$
and observe that one of its eigenvalues 
\begin{equation*}
    \frac{1}{2} q \left( (1 + p^2 q^2) (2 p - q) - 
   \sqrt{(1 + p^2 q^2)^2 (2 p - q)^2 + 
    4 p^2 (1 - p q)^2 ( (1 + p q)^2 - 2 q^2)  } \right) 
\end{equation*}
is negative for all $(p,q) \in \mathcal{R}$. Indeed, the eigenvalue is 
negative on $ \mathcal{R}$ if and only if 
\begin{IEEEeqnarray*}{C}
4 p^2 (1 - p q)^2 \left( (1 + p q)^2 - 2 q^2 \right) >0\\
\Longleftrightarrow \\
(1 + p q)^2 - 2 q^2 = (1 - p q)^2 + 2 (2 p-q)q >0,
\end{IEEEeqnarray*}
which proves that no $\Psi_{p,q}$ is co-completely positive.
\end{proof}

\begin{remark}
Since $\Psi_{p,q}$ and $\mathcal{Z}_{p,q}$ are  
invariant under the simultaneous sign change of $\pm p,\, \pm q$, \cref{thm:BS9} extends to the region in \cref{fig:leafpq} in the negative quadrant.
\end{remark}
\end{exmp}

\begin{exmp}[8 real zeros] \label{ex:BS8}
The last family of positive maps from~\cite{BucSiv} corresponds to $p_{\Psi}(\x, \y)$ with the following set of zeros,
\begin{multline} \label{eq:zerosmn}
 \mathcal{Z}_{m,n} = \left\{     (1, 1, e^{i \varphi};\,  1, 1, e^{i \varphi}), \,
       (1, -1, e^{i \varphi};\,  1, -1, e^{i \varphi}), \right. \\
 \left. (1,0,0;\, m,1,0),\, (1,n,0;\, 0,1,0),\, 
 (0,1,0;\, 0,0,1),\, (0,0,1;\, 1,0,0) \right\},
\end{multline}
where $m,n\in\RR$ and $\varphi \in [0, 2 \pi)$. 
In general 8 of these zeros are real,
\begin{multline*}
      (1,1,1; 1,1,1),\, (1,1,-1;1,1,-1),\, (1,-1,1;1,-1,1),\, (-1,1,1;-1,1,1), \\
(1,0,0;\, m,1,0),\, (1,n,0;\, 0,1,0),\, 
 (0,1,0;\, 0,0,1),\, (0,0,1;\, 1,0,0).
\end{multline*}

\begin{theorem} \label{thm:BS8}
Quantum maps $\Psi_{m,n} \colon \M_3^{\text{sa}} \to \M_3^{\text{sa}}$ 
of the form 
\begin{equation*}
 \begin{array}{c}
Z
\\
\begin{tikzcd}
  \arrow[mapsto]{d}
  \\
\phantom{a} 
\end{tikzcd}\\
\begin{bmatrix} \!\!
\substack{ n^2 \left(z_{00} + m (z_{01} + z_{10}) + m^2 z_{11} \right)}  
\!&\! 
\substack{ -m n (n z_{00} - z_{01} + m n z_{10} - m z_{11})} 
\!&\!
\substack{-n (m + n) (z_{02} + m z_{12})} \\
\substack{ -m n (n z_{00} - z_{10} + m n z_{01} - m z_{11})} 
\!&\! 
\substack{ m^2 \left(n^2 z_{00} - n (z_{01} + z_{10}) + z_{11} \right) } 
\!&\! 
\substack{  m (m + n) (n z_{02} - z_{12}) }\\
\!\! \substack{-n (m + n) (z_{20} + m z_{21})} 
\!&\! 
\substack{ m (m + n) (n z_{20} - z_{21}) } 
\!&\! 
\substack{(m + n)^2 z_{22} }
\!\! \end{bmatrix}\\
\\
+ b \begin{bmatrix}
    z_{11} & 0 & -z_{02} \\
    0 & z_{22} & -z_{12} \\
    -z_{20} & -z_{21} & z_{00} + z_{22}
\end{bmatrix}
+ c \begin{bmatrix}
    0& z_{01} - z_{10} & 0\\
    z_{10} - z_{01} & 0 & 0\\
    0 & 0 & 0
\end{bmatrix},
\end{array}   
\end{equation*}
where 
\begin{align*}
     b=& \min \set{-2 m n - m^2 n^2 - n^2,\, -2 m n - m^2 n^2 - m^2}, \\
     c=& \max \set{ -\frac{1}{2} m n (1 + m n - \sqrt{1 - n^2}),\, 
     -\frac{1}{2} m n (1 + m n - \sqrt{1 - m^2}) }
\end{align*}
are extreme positive maps in the region
\begin{displaymath}  \mathcal{A} =
\set{ (m,n) \in [-1,1]^2 \colon
-2 m n - m^2 n^2 - n^2 \geq  0, -2 m n - m^2 n^2 - m^2 \geq 0}.
\end{displaymath}
Apart from the boundary of $\mathcal{A}$, they are neither completely nor co-completely positive.
\end{theorem}

\begin{figure}[htbp]
  \centering
  \includegraphics[width=0.33\textwidth]{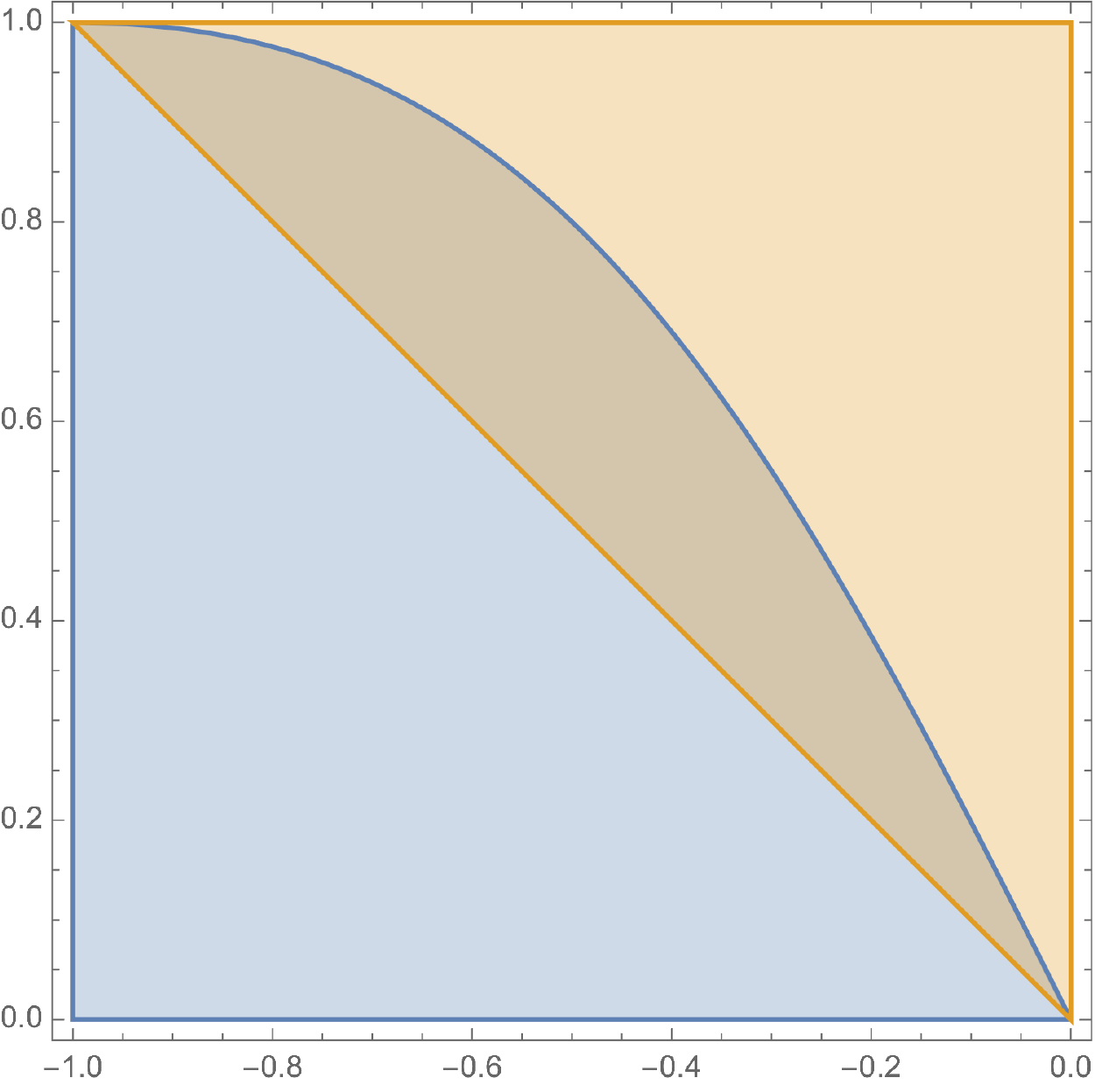}
   \hspace{8mm}
   \includegraphics[width=0.46\textwidth]{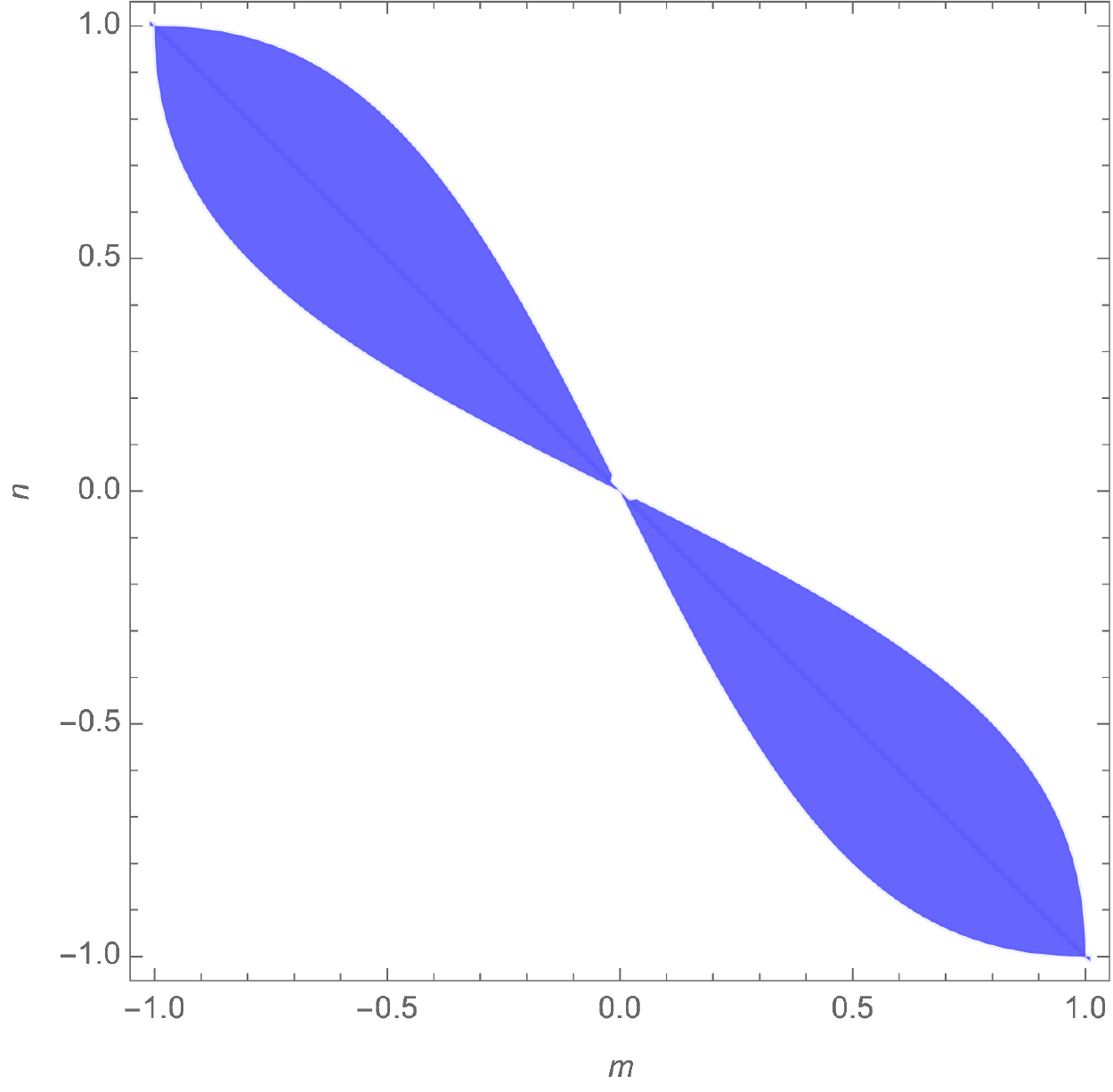}
   \put(-325,63){$\substack{n}$}
   \put(-260,-3){$\substack{m}$}
   \put(-237,90){$\substack{n \geq -m}$}
   \put(-300,40){$\substack{
   n \leq - \frac{2 m}{1 + m^2}}$}
  \caption{Region of points $(m,n) \in  \mathcal{A}$ in \cref{thm:BS8}
  for which $\Psi_{m,n}$ are positive maps.}
  \label{fig:leafmn}
\end{figure}

\begin{proof}
The symmetries of $(m,n)$ in the maps $\Psi_{m,n}$ and in the double leaf $\mathcal{A}$ permit us to restrict our study to a quarter of this region. 
Thus, from now on assume that 
$0\leq -m \leq n \leq 1$, as shown in \cref{fig:leafmn},
in which case $b=-2 m n - m^2 n^2 - n^2 $ and 
$c=-\frac{1}{2} m n (1 + m n - \sqrt{1 - n^2})$.

Using similar methods as in \cref{ex:BS9} proves that  $\Psi_{m,n}$ are positive and extremal maps on $\mathcal{A}$; see \cite{BucSiv} for the complete proof. 

The Choi matrix $C(\Psi_{m,n})$ is 
\begin{equation*}
 \left[   \begin{array}{ccc|ccc|ccc}
\substack{n^2} & \substack{m n^2}  &  \cdot  &  
\substack{-m n^2}  & \substack{m n+c} & \cdot  & 
\cdot  &  \cdot   & \substack{m n + m^2 n^2}   \\
 \substack{m n^2}   & \substack{-2 m n - n^2} &  \cdot & 
 \substack{-m^2 n^2-c} &  \substack{m^2 n}  &  \cdot  &
   \cdot  &  \cdot  &  \substack{-m n (m + n)}  \\
 \cdot  &  \cdot  & \substack{0} &
  \cdot  &  \cdot  &  \cdot & \cdot & \cdot  &  \cdot  \\ 
\hline 
 \substack{-m n^2}  &\substack{-m^2 n^2 -c} &  \cdot  &
 \substack{m^2 n^2}  & \substack{-m^2 n}  &  \cdot  &
  \cdot  &  \cdot  &  \substack{m n (m + n)}  \\
\substack{m n+c} & \substack{m^2 n}  &  \cdot  &
 \substack{-m^2 n} & \substack{m^2} & \cdot &
  \cdot  &  \cdot  & \substack{-m (m + n) -b}  \\
 \cdot  &  \cdot &  \cdot  &
  \cdot  &  \cdot  & \substack{b} &
   \cdot & \cdot & \cdot \\
\hline
 \cdot  &  \cdot  & \cdot &
  \cdot  &  \cdot  &  \cdot &
\substack{b}  &  \cdot  &  \cdot \\
 \cdot  &  \cdot  &  \cdot  &
  \cdot  &  \cdot  & \cdot & 
   \cdot & \substack{0} &  \cdot \\
\substack{m n  + m^2 n^2} & \substack{-m n (m + n)}  &  \cdot  &
 \substack{m n (m + n)}  & \substack{-m (m + n) -b} &  \cdot  &
  \cdot &  \cdot  & \substack{m^2 - m^2 n^2} 
\end{array} \right].
\end{equation*}
From the shape of the matrix $C(\Psi_{m,n})$ we can read
its 9 eigenvalues: $2 \times 0,\, 2 \times b $ 
and the five roots of the characteristic polynomial of the  remaining
$5 \times 5$ minor.
The characteristic polynomial is of the form
\begin{equation*}
    C_0+ C_1 \#+ C_2 \#^2+ \frac{1}{2} b (2 m^2 + 4 n^2 + 7 m^2 n^2) \#^3 +4m (n-m) \#^4 + \#^5,
\end{equation*}
where $C_0,C_1,C_2$ are simple functions of $m$ and $n$.
An explicit plot of its roots in \texttt{Mathematica}, shown in \cref{fig:eigenval}, verifies that 
all the roots are real and only one root attains negative value on $\mathcal{A} \backslash \partial \mathcal{A}$. 
\begin{figure}[htbp]
  \centering
  \includegraphics[width=0.7\textwidth]{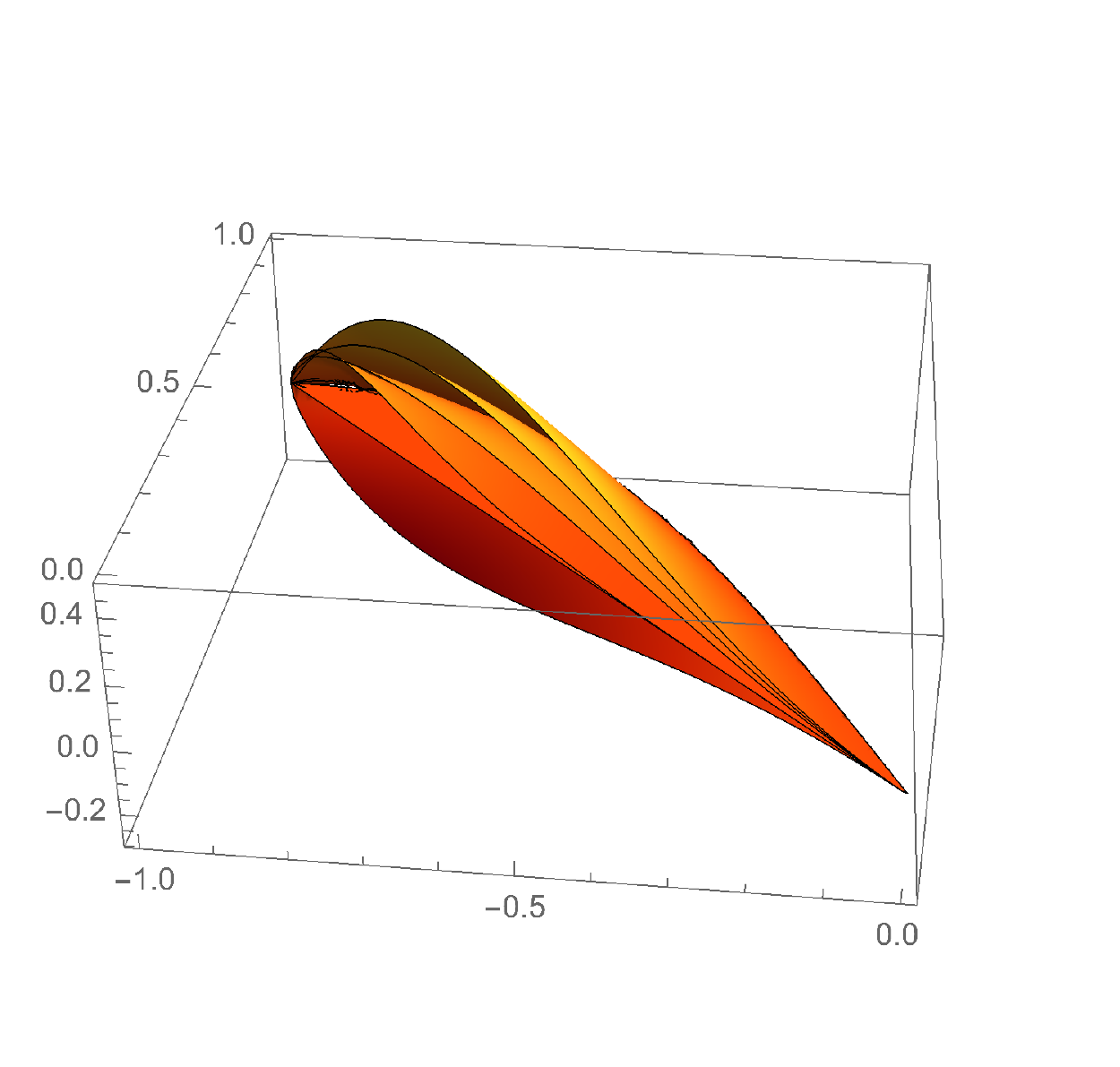}
   \includegraphics[width=0.075\textwidth]{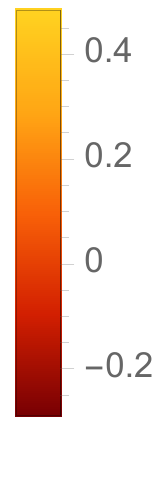}
   \put(-270,133){$\substack{n}$}
   \put(-135,30){$\substack{m}$}
  \caption{Eigenvalues of $C(\Psi_{m,n})$ as functions of $m,n$: only one eigenvalue is negative.}
  \label{fig:eigenval}
\end{figure}
This implies that on the boundary of $\mathcal{A}$, where $b=c=0$,
the maps $\Psi_{m,n}$ are completely positive. Moreover, 
for $(m,n) \in \partial \mathcal{A}$,
\begin{equation} \label{pg:mnCP}
   \Psi_{m,n}(Z)=
     \begin{bmatrix}
        -n & -m n & 0 \\
        m n & -m & 0 \\
        0 & 0 & m+n
    \end{bmatrix}
      Z 
    \begin{bmatrix}
        -n & m n & 0 \\
        -m n & -m & 0 \\
        0 & 0 & m+n
    \end{bmatrix}
\end{equation}
is a completely positive map with the Kraus rank 1. 

It remains to be checked if any $\Psi_{m,n}$ is co-completely positive.
To this end we compute the Choi matrix 
$C\left( T \circ \Psi_{m,n} \right)$
and observe that two of its eigenvalues 
are negative for all $(m,n) \in \mathcal{A} \backslash 
\set{(0,0), (-1,1), (1,-1)}$; thus in these points the maps
 $\Psi_{m,n}$ are not 
co-completely positive. At $(m,n)=(0,0)$ the map $\Psi_{0,0}$ is constantly zero. When $(m,n)$ is $(-1,1)$ or $(1,-1)$, we get the following rank 1 
maps,
\begin{equation*}
\Psi_{\pm 1,\mp 1}(Z) = T \circ \Psi_{\pm 1,\mp 1}(Z) =
( z_{00} \pm z_{01} \pm z_{10} + z_{11})  \begin{bmatrix}
       1 & \mp 1 & 0 \\
         \mp 1 & 1 & 0 \\
           0 & 0 & 0 \\
    \end{bmatrix}.
\end{equation*}
\end{proof}

\begin{remark}
Note that the positive maps $\Psi_{m,n}$ are very different from the maps $\Psi_t$ and $\Psi_{p,q}$ in
\cref{thm:BS10,thm:BS9}. Most notably,
$\Psi_{m,n} \colon \M_3^{\text{sa}} \to \M_3^{\text{sa}}$
is not a trivial extension of its restriction to real symmetric matrices, i.e., 
$\Psi_{m,n} \colon \M_3^{\text{sym}} \to \M_3^{\text{sym}}$.
\end{remark}

\end{exmp}

\subsection{Unital entanglement witnesses}
\label{subsec:unital}
In quantum information theory it is convenient to require that 
an entanglement witness $\Psi$ is \textit{unital}, i.e., $\Psi(I)=I$. 
By the following lemma we can make the examples in \cref{subsec:BucSiv} unital. 

\begin{lemma} \label{ex:UnEW}
Horodecki's entanglement witness 
theorem 
remains valid if we require that positive maps
are unital. 
\end{lemma}

\begin{proof}
If $\rho$ and $\Phi$ are as in \Cref{cor:HEW}, then the map 
$\hat{\Phi}(X) := \Phi(X) + \varepsilon\, (\Tr X)\, I $ also 
fulfills corollary's conclusions for small 
enough $\varepsilon>0$. Consequently, the map 
\begin{equation*}
    \tilde{\Phi} (X):= \hat{\Phi}(I)^{-1/2}\, \hat{\Phi}(X)\,  \hat{\Phi}(I)^{-1/2}
\end{equation*}
is unital and satisfies the properties of Horodecki's entanglement witness 
theorem.
Moreover, 
when $\Phi(I)$ is positive definite (i.e., invertible), we can choose $\varepsilon=0$ and then $\tilde{\Phi}$ is an example of 
the $\text{PGL}_3 \times \text{PGL}_3$ action on the cone of positive maps $\bP(\M_3^{\text{sa}}, \M_3^{\text{sa}})$ considered in
\eqref{eq:PQaction}.
\end{proof}

It is also useful for positive maps to be \textit{trace preserving}, i.e., $\Tr \Psi(\rho)=\Tr \rho$ for all $\rho \in \M^{\text{sa}}_3$. It is easily verified that a quantum map $\Psi$ is
trace preserving if and only if $\Psi^{\dagger}$  is  unital.
We invite the reader to consult 
\cite[\S 2.4]{AandBbook} on how this can be achieved.

In the rest of this subsection we apply \cref{ex:UnEW}
to the families of positive maps in \cref{thm:BS10,thm:BS9,thm:BS8}.
Let $\Psi_t$ be the positive map from \cref{thm:BS10}. 
Then it is easily checked that the map 
$\frac{1}{ 2 \left(1 - t^2 + t^4 \right)}\, \Psi_t$ is  unital and trace preserving. 

The map $\Psi_{p,q}$ from \cref{thm:BS9} is  neither unital nor trace preserving. 
Since 
\begin{equation*} \Psi_{p,q}(I)=
    \begin{bmatrix}
    \substack{p^2 (1-p q)^2 + (2 p-q) q}    & 0 & 0\\
    0 &  
    \substack{2 p q \left( (1-p q)^2 +(2 p-q)q \right)} & 0\\
  0 & 0 &
\substack{ p^2 (1-p q)^2 + (2 p-q) q}  
\end{bmatrix}
\end{equation*}
is positive definite for all $(p,q) \in \mathcal{R}$, we can set $\varepsilon=0$ in the proof of \cref{ex:UnEW}. 
Then the newly obtained unital map has the form 
\begin{equation} \label{eq:pqUnital}
    \Psi_{p,q}^{\text{u}} (Z):= \Psi_{p,q}(I)^{-1/2}\, \Psi_{p,q}(Z)\,  \Psi_{p,q}(I)^{-1/2}.
\end{equation}
In particular, on the segment $q=2p$ on $\mathcal{R}$, 
the Krause rank 1 of the completely positive maps $\Psi_{p,2p}$ is preserved. Moreover, from 
\begin{equation*}
\Psi_{p,2p}(I)^{-1/2} =     \frac{1}{p (1-2 p^2)}
     \begin{bmatrix}
        1 & 0 & 0 \\
        0 & 1/2 & 0 \\
        0 & 0 & 1
    \end{bmatrix}
\end{equation*}
it follows that $\Psi_{p,2p}^{\text{u}}$ is independent on $p$:
\begin{equation*}
\Psi_{p,2p}^{\text{u}} (Z) =
     \begin{bmatrix}
        1 & 0 & 0 \\
        0 & -1 & 0 \\
        0 & 0 & 1
    \end{bmatrix}
      Z 
    \begin{bmatrix}
        1 & 0 & 0 \\
        0 & -1 & 0 \\
        0 & 0 & 1
    \end{bmatrix}.
\end{equation*}

Also the map $\Psi_{m,n}$ in \cref{thm:BS8} is  neither unital nor trace preserving. 
First observe that 
\begin{equation*} \Psi_{m,n}(I)=
    \begin{bmatrix}
    \substack{n^2+ m^2 n^2+b}    & \substack{m n (m-n)} & 0\\
     \substack{m n (m-n)} &   \substack{m^2+ m^2 n^2+b}   & 0\\
  0 & 0 &
\substack{ (m+n)^2 +2b}  
\end{bmatrix}
\end{equation*}
is positive definite for all 
$(m,n) \in \mathcal{A} \backslash \set{(0,0),(-1,1),(1,-1)}$
(recall $b\geq 0$ on $\mathcal{A}$ and 0 on $\partial \mathcal{A}$),
because
\begin{equation*}
    \det \left( \Psi_{m,n}(I) \right)=
\left(2 b + (m+n)^2 \right) 
\left(b^2+b (m^2 + n^2+2  m^2 n^2) + m^2 n^2 (1 + m n )^2 \right).
\end{equation*}
Therefore we can again set $\varepsilon=0$ and define
the unital map  as
\begin{equation} \label{eq:mnUnital}
    \Psi_{m,n}^{\text{u}} (Z):= \Psi_{m,n}(I)^{-1/2}\, \Psi_{m,n}(Z)\,  \Psi_{m,n}(I)^{-1/2} .
\end{equation}
Consequently, for $(m,n) \in \partial \mathcal{A}$
we get
\begin{equation*}
\Psi_{m,n}(I)^{-1} =     
     \begin{bmatrix}
\frac{1+n^2}{n^2 (1+m n)^2}  & \frac{n-m}{m n (1+m n)^2} & 0 \\
\frac{n-m}{m n (1+m n)^2} & \frac{1+m^2}{m^2 (1+m n)^2} & 0 \\
        0 & 0 & \frac{1}{(m + n)^2} 
    \end{bmatrix},
\end{equation*}
and by \cref{pg:mnCP} the completely positive maps 
$\Psi^{\text{u}}_{m,n}$
on the boundary are  equal to
\begin{equation*}
\Psi_{m,n}^{\text{u}} (Z) =\Psi_{m,n}(I)^{-1/2}
     \begin{bmatrix}
        \substack{-n} & \substack{-m n} & \substack{0} \\
        \substack{m n} & \substack{-m} & \substack{0} \\
       \substack{ 0} & \substack{0} & \substack{m+n}
    \end{bmatrix}
      Z 
    \begin{bmatrix}
        \substack{-n} & \substack{m n} & \substack{0} \\
        \substack{-m n} & \substack{-m} & \substack{0} \\
       \substack{ 0} & \substack{0} & \substack{m+n}
    \end{bmatrix}
    \Psi_{m,n}(I)^{-1/2}.
\end{equation*}

\section{Algorithm and its analysis}
\label{sec:alg}
The summary of the construction of a new entanglement witness 
$\Psi$ or $C(\Psi)$
in \cref{sec:examples}  is as follows: 
\begin{equation*}
    C(\Psi) \notin \text{co-}\!\PSD + \PSD  \Longleftrightarrow
    \text{ there exists } \rho \in \mathcal{PPT} 
    \text{ such that } 
     \langle C(\Psi), \rho
    \rangle_{\text{HS}} <0.
\end{equation*}
Such an entanglement witness induces a linear functional 
$\langle C(\Psi), \, \cdot \,  \rangle_{\text{HS}}$. 
The next proposition rephrases the above equivalence 
as an optimization problem. 

\begin{proposition} Let $\Psi \colon \M_3^{\text{sa}} \rightarrow 
\M_3^{\text{sa}}$ be a 
positive map which is not decomposable.
The solution of the semidefinite program
\begin{IEEEeqnarray*}{CC} 
\text{minimize}\colon & \Tr \left( C(\Psi)\, \rho \right)\\
\text{subject to}\colon & \rho^{\Gamma} \succeq 0 \\
    & \rho \succeq 0
\end{IEEEeqnarray*}
is an entangled PPT state on $\CC^3 \otimes \CC^3$. 
\end{proposition}

\Cref{alg:SDP} takes $\Psi \in \bP \backslash \boldsymbol{DEC}$ as 
input and the output is an entangled state $\rho \in \M_9^{\text{sa}}$
whose entanglement cannot be detected by  
$\Gamma=T \otimes \Id$.

\begin{algorithm}
\caption{Semidefinite program}
\label{alg:SDP}
\begin{algorithmic} \vspace{2mm}
\STATE{\hspace{2cm} \textbf{minimize:} \hspace{5mm} 
$\Tr \left(C(\Psi)\, \rho \right)$}\vspace{2mm}
\STATE{\hspace{18.5mm} \textbf{subject to:} \hspace{8mm} 
$\rho^{\Gamma} \succeq 0 $}\vspace{1mm}
\STATE{\hspace{49.5mm}  
$\rho \succeq 0 $}\vspace{1mm}
\end{algorithmic}
\end{algorithm}

Furthermore, we wish to construct density matrices that are not detected by the Choi map \eqref{eq:ChoisMap}. This imposes a new constraint to the semidefinite program, namely  
$(\Psi_C^{\dagger} \otimes \Id) \rho \succeq 0$ as shown 
in \cref{fig:PT_Choi}.
By taking into account this additional constraint, \cref{alg:SDPandChoi} returns states which are detected neither by the transposition nor by the Choi map. See \cref{fig:Psi_t} for the graphical description of \cref{alg:SDPandChoi}.

\begin{algorithm}
\caption{Semidefinite program}
\label{alg:SDPandChoi}
\begin{algorithmic} \vspace{2mm}
\STATE{\hspace{2cm} \textbf{minimize:} \hspace{5mm} 
$\Tr \left(C(\Psi)\, \rho \right)$}\vspace{2mm}
\STATE{\hspace{18.5mm} \textbf{subject to:} \hspace{2mm} 
$(\Psi_C^{\dagger} \otimes \Id) \rho \succeq 0$ } \vspace{1mm}
\STATE{\hspace{49.5mm}  $\rho^{\Gamma} \succeq 0 $}\vspace{1mm}
\STATE{\hspace{50.5mm}  
$\rho \succeq 0 $}\vspace{1mm}
\end{algorithmic}
\end{algorithm}

\begin{figure}[htbp]
  \centering
  \includegraphics[width=0.73\textwidth, height=95mm]{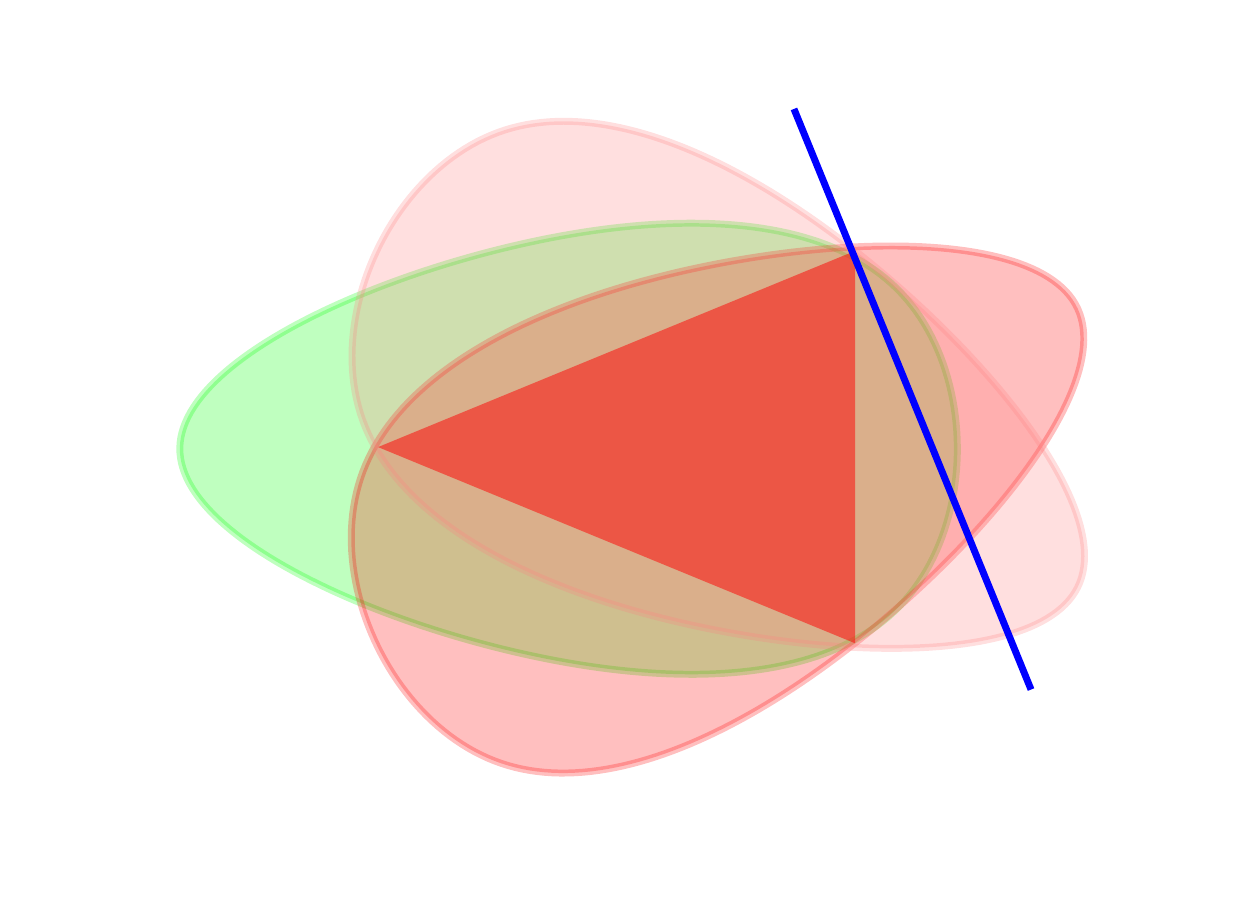}
   \put(-160,215){$\Gamma(\D)$}
    \put(-135,134){$\Sep$ }
    \put(-250,134){$( \Psi_C^{\dagger} \otimes \Id )(\D)$ }
    \put(-160,50){$\D$}
         \put(-75,61) {\vector(0,1){61}}
         \put(-82,61) {\vector(-1,1){23}}
       \put(-102,52){\color{blue} $\langle 
        C(\Psi), \cdot \,  \rangle_{\text{HS}} > 0$}
       \put(-100,40){entangled density matrices }
       \put(-100,30){not detected by $\Psi$ }
       \put(-76,205) {\vector(0,-1){32}}
       \put(-82,233){\color{blue} \boldmath$\langle 
        C(\Psi), \cdot \,  \rangle_{\text{HS}} < 0$}
            \put(-80,220){entangled density matrices  }
             \put(-80,210){detected by $\Psi$}
             \put(-71,166){\color{blue} $\bullet \rho$}
  \caption{Entanglement witness $\Psi$ as a linear functional 
 $\langle C(\Psi),\, \cdot \,  \rangle_{\text{HS}}$. Density matrix $\rho$ is such that $\langle C(\Psi), \rho  \rangle_{\text{HS}} =
 \Tr \left(C(\Psi)\, \rho \right)< 0$ is minimal.}
  \label{fig:Psi_t}
\end{figure}

Before we run our semidefinite program on positive maps in
\cref{thm:BS10,thm:BS9,thm:BS8}, we prove that it suffices to 
execute it for real symmetric matrices.

 \begin{lemma} \label{lem:sym}
 When the linear objective function in \cref{alg:SDPandChoi} is independent on the imaginary part of 
 $\rho \in \M_9^{\text{sa}}$,
 it suffices to run the algorithm for real symmetric matrices with trace 1.
 \end{lemma}

\begin{proof}
Consider a solution $\rho = [s_{ij}]_{i,j=0}^8 \in \M_9^{\text{sa}}$ of \cref{alg:SDPandChoi}.
First observe that $ \frac{1}{2} \left(\rho+\ols{\rho}\right)= 
    \frac{1}{2} \left(\rho+\rho^T\right) \in \M_9^{\text{sym}}$
is a real symmetric density matrix which, under the assumptions of the lemma, achieves the same minimal value 
$\Tr \left(C(\Psi)\, \rho \right)$.     
Next we explicitly write the constraints in \cref{alg:SDPandChoi} for 
$\rho  \in \M_9^{\text{sa}}$:

\begin{displaymath}
\rho^{\Gamma} = (T \otimes \Id) \rho \succeq 0,
\end{displaymath}
\begin{multline*}
  \hspace{5.3cm}  (\Psi_C^{\dagger} \otimes \Id) \rho \succeq 0 \\
    || \\
\left[\! \begin{array}{ccc|ccc|ccc}
\substack{s_{00}+s_{66}} & \substack{s_{01} + s_{67}}  & \substack{s_{02} + s_{68}} & 
-s_{03} & - s_{04} & -s_{05} &
-s_{06} & - s_{07} & -s_{08}\\
\substack{\ols{s_{01}}+\ols{s_{67}}} & \substack{s_{11} + s_{77}}  & \substack{s_{12} + s_{78}} & 
-s_{13} & - s_{14} & -s_{15} &
-s_{16} & - s_{17} & -s_{18}\\
\substack{\ols{s_{02}}+\ols{s_{68}}} & \substack{\ols{s_{12}}+\ols{s_{78}} } & \substack{ s_{22} + s_{88} } & 
-s_{23} & - s_{24} & -s_{25} &
-s_{26} & - s_{27} & -s_{28} \vspace{1mm} \\
\hline 
-\ols{s_{03}} & -\ols{s_{13}} & -\ols{s_{23}} & 
\substack{s_{00}+s_{33} } & \substack{s_{01}+s_{34}} & \substack{s_{02}+s_{35}} &
-s_{36} & - s_{37} & -s_{38}\\
-\ols{s_{04}} & -\ols{s_{14}} & -\ols{s_{24}} & 
\substack{\ols{s_{01}} + \ols{s_{34}}} & \substack{s_{11}+s_{44} } & \substack{s_{12}+s_{45}}  &
-s_{46} & - s_{47} & -s_{48}\\
-\ols{s_{05}} & -\ols{s_{15}} & -\ols{s_{25}} & 
\substack{\ols{s_{02}}+\ols{s_{35}}} & \substack{\ols{s_{12}}+\ols{s_{45}}}  &\substack{ s_{22}+s_{55}} &
-s_{56} & - s_{57} & -s_{58} \vspace{1mm}\\
\hline 
-\ols{s_{06}} & -\ols{s_{16}} & -\ols{s_{26}} & 
-\ols{s_{36}} & -\ols{s_{46}} & -\ols{s_{56}} & 
\substack{s_{33}+s_{66}} & \substack{s_{34}+s_{67}} & \substack{s_{35}+s_{68}} \\ 
-\ols{s_{07}} & -\ols{s_{17}} & -\ols{s_{27}} & 
-\ols{s_{37}} & -\ols{s_{47}} & -\ols{s_{57}} & 
\substack{\ols{s_{34}}+\ols{s_{67}}} &\substack{ s_{44}+s_{77}} & \substack{s_{45}+s_{78}}\\
-\ols{s_{08}} & -\ols{s_{18}} & -\ols{s_{28}} & 
-\ols{s_{38}} & -\ols{s_{48}} & -\ols{s_{58}} & 
\substack{\ols{s_{35}}+\ols{s_{68}}} & \substack{\ols{s_{45}}+\ols{s_{78}}} & \substack{s_{55}+s_{88}}
\end{array} \!\right]\!\!.
\end{multline*}
Since the transposition and the Choi map commute with the transposition
(which can also be seen from the above matrix), it holds that
\begin{IEEEeqnarray*}{CCCC} 
(T\otimes \Id)\rho^T &=& \left((T\otimes \Id)\rho \right)^T, & \\
(\Psi_C^{\dagger} \otimes \Id)\rho^T &=& 
\left( (\Psi_C^{\dagger} \otimes \Id)\rho \right)^T & \text{ and }\\
\Tr \left(C(\Psi)\, \rho \right) &=& 
\Tr \left(C(\Psi)\, \rho^T \right). & 
\end{IEEEeqnarray*} 
Therefore $\rho \in \M_9^{\text{sa}}$ is an instance of \cref{alg:SDPandChoi} if and only if $\rho^T$ is an instance of the same algorithm.
Moreover, the same is true for the real symmetric state $\frac{1}{2} \left(\rho +\rho^T\right)$. 
 
Together with the trace condition 
$\Tr \rho = \sum_{i=0}^8 s_{ii} =1$,
we can execute the semidefinite program in \cref{alg:SDPandChoi} on a convex compact set in 
 $\M_9^{\text{sym}} \cong \RR^{45}$.
\end{proof}

The next lemma explores how \cref{alg:SDP,alg:SDPandChoi} 
change if we make entanglement witnesses unital, like in the examples in 
\cref{subsec:unital}. 
\begin{lemma} \label{lem:projA}
Let $\Psi \colon \M_3^{\text{sa}} \rightarrow \M_3^{\text{sa}}$ be a positive map and $Q \in \M_3$. Then $\Psi_Q$ defined by 
$Z  \mapsto  Q \, \Psi(Z) Q^{\dagger}$ is another positive map, and 
their Choi matrices are connected as 
\begin{equation*}
    C(\Psi_Q) =   (Q \otimes I)\,  C(\Psi)
    \,  (Q \otimes I)^{\dagger}.
\end{equation*}
\end{lemma}

\begin{proof}
For a product state $\rho=\tau \otimes \tau'$ it holds
\begin{equation*}
    (\Psi_Q \otimes \Id) \rho = 
    Q \Psi(\tau) Q^{\dagger} \otimes \tau' =
    (Q \otimes I)\, \left( (\Psi \otimes \Id) \rho  \right)\,  
    (Q^{\dagger} \otimes I),
\end{equation*}
and by linearity the same holds for any state. 
The relation between $C(\Psi) $ and $C(\Psi_Q) $ then
easily follows from the Choi matrix representation \cref{eq:ChoiMat}. 
\end{proof}

The relation between the Choi matrices of $\Psi$ and $\Psi_Q$ implies the relation between the objective functions in \cref{alg:SDP,alg:SDPandChoi}:
\begin{equation*}
    \Tr \left(C(\Psi_Q)\, \rho \right)= 
     \Tr \left(   C(\Psi)
    \,  (Q^{\dagger} \otimes I)  \rho (Q \otimes I) \right).
\end{equation*}
Next we prove that,
when $Q$ is invertible,  the constraints of \cref{alg:SDP} 
remain invariant under the action 
$ \rho \mapsto (Q^{\dagger} \otimes I)\,  \rho \, (Q \otimes I)$.
Indeed, positive semidefiniteness is 
clearly preserved. Moreover,  $\rho^{\Gamma} \succeq 0$ if and only if
$\left( (Q^{\dagger} \otimes I)\,  \rho \, (Q \otimes I) \right)^{\Gamma} \succeq 0$, which follows by linearity from the below calculation 
on a product state $\rho = \tau \otimes \tau'$,
\begin{multline*}
(T \otimes \Id) 
\left( (Q^{\dagger} \otimes I)\,  \rho \, (Q \otimes I) \right)
 =
(T \otimes \Id) 
\left( Q^{\dagger} \tau Q  \otimes \tau' \right) = 
 Q^T \tau^T \Bar{Q}  \otimes \tau' = \\
   (Q^T \otimes I)\,  (\tau^T \otimes \tau') \, (\Bar{Q} \otimes I) = 
   \left(Q^T \otimes I \right)\,  \rho^{\Gamma} \, 
  \left( (Q^T)^{\dagger} \otimes I \right).
\end{multline*}
On the contrary, the constraint
in \cref{alg:SDPandChoi} involving the Choi map $\Psi_C$
is not invariant  under the action 
$ \rho \mapsto (Q^{\dagger} \otimes I)\,  \rho \, (Q \otimes I)$.

\section{Experimental results: New examples of entangled density matrices}
\label{sec:experiments}
In this section we execute \cref{alg:SDPandChoi} for 
entanglement witnesses in \cref{thm:BS10,thm:BS9,thm:BS8}. As explained in
\cref{sec:examples,sec:alg}, for
positive maps that are not decomposable, the outputs will be new entangled states.

\begin{exmp}
The map $\frac{1}{ 2 \left(1 - t^2 + t^4 \right)}\, \Psi_t$ from \cref{thm:BS10} is  unital and trace preserving. From now on, for the sake of shorter notation, we denote it by $\Psi_t$.

In \cref{alg:SDPandChoi} for $\Psi_t$ we are minimizing 
\begin{multline} \label{eq:trace_t}
\Tr \left(C(\Psi_t)\, \rho \right) = \\
\frac{1}{ 2 \left(1 - t^2 + t^4 \right)} \left( s_{11} + s_{55} + s_{66} +
    t^4  (s_{22} + s_{33} + s_{77}) +
   (1 - t^2)^2 (s_{00} + s_{44} + s_{88})  - \right. \\
 \left.   (1 - t^2 + t^4) (s_{04} + \ols{s_{04}} + 
   s_{08} + \ols{s_{08}}+
   s_{48} + \ols{s_{48}})  \right),
\end{multline}
which does not depend on the imaginary part of $\rho = [s_{ij}]_{i,j=0}^8$,
thus \cref{lem:sym} applies.
We execute \cref{alg:SDPandChoi} in \texttt{Mathematica} (which permits semidefinite programming only on real symmetric matrices) for $\Psi_t$ for 2001 equidistant values $t \in [-10, 10]$. The output values 
$\Tr \left(C(\Psi_t)\, \rho_t \right)$ are shown in \cref{fig:tracesPsi_t}.
 \begin{figure}[htbp]
  \centering
  \includegraphics[width=0.73\textwidth, height=50mm]{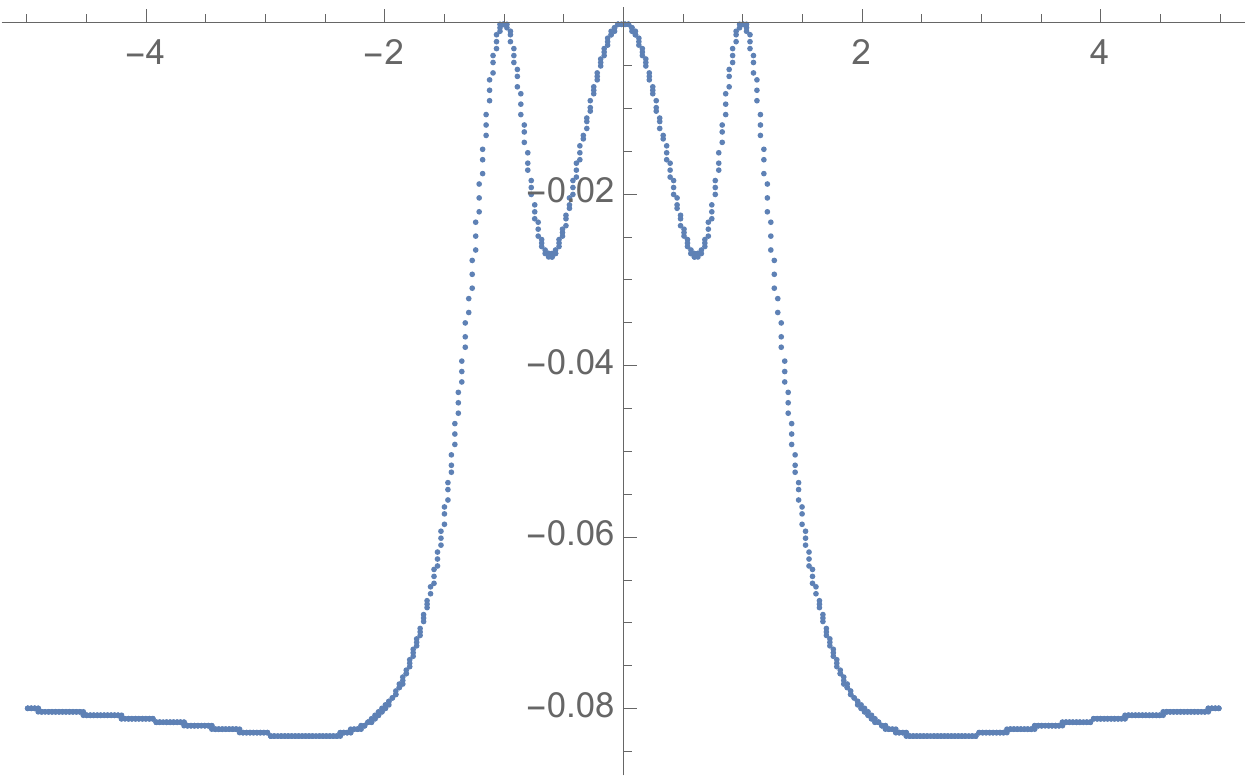}
  \put(3,135){$t$}
        \put(-170,-8){
        $\Tr \left(C(\Psi_t)\, \rho_t \right)$}
  \caption{Plot of outputs of \cref{alg:SDPandChoi}:
  $\left( t, 
 \Tr \left[C(\Psi_t)\, \rho_t \right] \right)$  for $t\in[-10,10]$.}
  \label{fig:tracesPsi_t}
\end{figure}
The output density matrices are of the sparse form 
\begin{equation*} \label{eq:rho_t}
 \rho_t = \left[   \begin{array}{ccc|ccc|ccc}
s_{00} & \cdot &  \cdot  &  
\cdot  & s_{04} & \cdot  & 
\cdot  &  \cdot   & s_{04}  \\
 \cdot  & s_{11} &  \cdot & \cdot &  \cdot  &  \cdot  &
   \cdot  &  \cdot  &  \cdot  \\
 \cdot  &  \cdot  &  s_{22} &
  \cdot  &  \cdot  &  \cdot & \cdot & \cdot  &  \cdot  \\
\hline
 \cdot  &\cdot &  \cdot  &
 s_{22} &  \cdot  &  \cdot  &
  \cdot  &  \cdot  &  \cdot  \\
\ols{s_{04}} &   \cdot  &  \cdot  &
 \cdot & s_{00} & \cdot &
  \cdot  &  \cdot   & s_{04}  \\
 \cdot  &  \cdot &  \cdot  &
  \cdot  &  \cdot  & s_{11} &
   \cdot & \cdot & \cdot \\
\hline
 \cdot  &  \cdot  & \cdot &
  \cdot  &  \cdot  &  \cdot &
  s_{11}  &  \cdot  &  \cdot \\
 \cdot  &  \cdot  &  \cdot  &
  \cdot  &  \cdot  & \cdot & 
   \cdot & s_{22} &  \cdot \\
\ols{s_{04}} & \cdot  &  \cdot  &
 \cdot  & \ols{s_{04}} &   \cdot  &
  \cdot &  \cdot  & s_{00}
\end{array} \right].
\end{equation*}

Since $\Tr \left(C(\Psi_t)\, \rho \right)$ is linear, the minimum is achieved on the boundary of the constraint domain.
If we exploit the symmetries and the sparse form of $\rho_t$,
\begin{IEEEeqnarray*}{L} 
s_{11} = s_{55} =  s_{66},\ \  
s_{22} =  s_{33} = s_{77}, \\
s_{00} =  s_{44} =  s_{88}, \ \
s_{04} =  s_{48} =  s_{08} \ \text{ and }\\
s_{00}+ s_{11}+s_{22}= \frac{1}{3}, 
\end{IEEEeqnarray*} 
we can run \cref{alg:SDPandChoi} symbolically by using the method of Lagrange multipliers.  
The constraints command 
$\rho_t$, 
$(T\otimes \Id)\rho_t$ and
$(\Psi_C^{\dagger} \otimes \Id)\rho_t$ to be positive semidefinite.
The respective eigenvalues are: 
\begin{align*}
    \rho_t \colon & 
    s_{00}-s_{04}, s_{00}-s_{04}, s_{00}+ 2 s_{04},\,
    s_{11},s_{11}, s_{11},\,
    s_{22}, s_{22}, s_{22}\, ; \\
     \rho_t^{\Gamma} \colon &
     3 \times s_{00},\, 3 \times 
\frac{1}{2} \left(s_{11} + s_{22} \pm 
\sqrt{4 s_{04}^2 + (s_{11} - s_{22})^2} \right) \, ;\\   
(\Psi_C^{\dagger} \otimes \Id)\rho_t \colon &
s_{00}- 2 s_{04} + s_{11},\, 2 \times s_{00} + s_{04} + s_{11}, \,
3 \times s_{00} + s_{22},\, 
3 \times s_{11} + s_{22}\, .
\end{align*}

Together with the trace constraint $s_{00}=\frac{1}{3}-s_{11}-s_{22}$,
\cref{alg:SDPandChoi} translates into the following optimization 
problem: \vspace{2mm}

\textbf{minimize:}
\begin{equation*}
\frac{(1 - t^2)^2 + 3  t^2 (2 - t^2) s_{11} -
   3 (1-2 t^2) s_{22}  - 6 (1 - t^2 + t^4)s_{04}}
{2(1 - t^2 + t^4)} 
\end{equation*}  

\textbf{subject to:}
\begin{IEEEeqnarray*}{C}
0 \leq s_{11} \leq \frac{1}{3},\  0 \leq s_{22} \leq \frac{1}{3},\ 
-\frac{1}{3} \leq s_{04} \leq \frac{1}{3}, \\
s_{11} s_{22} - s_{04}^2 \geq 0, \\
\frac{1}{3}- s_{11} - s_{22} + 2 s_{04} \geq 0,\  
\frac{1}{3} - s_{11} - s_{22} - s_{04} \geq 0, \ 
 \frac{1}{3}  - s_{22} - 2 s_{04}  \geq 0.
\end{IEEEeqnarray*}
The solutions are displayed in \cref{tab:rho_t}.

\begin{table}[htbp]
{\footnotesize
  \caption{All the values of $\Tr \left(C(\Psi_t)\, \rho_t \right)$ are nonpositive}  \label{tab:rho_t}
\begin{center}
  \begin{tabular}{|c|c|c|} \hline 
   $t\in\RR$ & $\min \Tr \left(C(\Psi_t)\, \rho_t \right) $ & 
   $\rho_t$  \\ \hline 
    & & \\
       &  & $s_{11}=
       \frac{2}{9}  -  \frac{1}{3 \sqrt{3 + 3 t^4 + 3 t^8}}  $ \\
   $|t| \geq 1$ &  
   $\frac{3 (1 + t^4) - 2 \sqrt{3 + 3 t^4 + 3 t^8}}{6 (1 - t^2 + t^4)}$ 
   & $s_{22}=
       \frac{2}{9}  -  \frac{t^4}{3 \sqrt{3 + 3 t^4 + 3 t^8}}  $ \\
        &    & 
   $s_{00}=s_{04}=\frac{1}{3}-s_{11}-s_{22}$ \\
    & & \\   \hline
    & & \\
              &  & 
$ s_{11}=
\frac{ 4 + t^2}{6 \sqrt{ 12 - 3t^4}} - \frac{1}{6} $ \\ 
 $|t| \leq 1$  &
 $\frac{t^2 (t^2 - 4 + \sqrt{12 - 3 t^4})}{4 (1 - t^2 + t^4)}$ & 
 $ s_{22}= 
 \frac{\sqrt{12 - 3 t^4}}{9  (2 + t^2)}$  \\ 
& & $s_{00}=\frac{1}{3}-s_{11}-s_{22}$ \\
& & $s_{04}=\frac{1}{2} (s_{00}+s_{11})$ \\
& & \\ \hline
  \end{tabular}
\end{center}
}
\end{table}

\begin{remark} Note that for $t=0$ or $t =\pm 1$, the output minimum value of \cref{alg:SDPandChoi} is zero. This means that for these $t$ the algorithm does not return new entangled sates. This is as expected, since $\Psi_0$ is the Choi map used in the constraints and $T\circ \Psi_{\pm 1}$ are completely positive. \Cref{tab:rho_t} shows that \cref{alg:SDPandChoi} actually returns separable states $\rho_{\pm 1}$ and  $\rho_0$.
\end{remark}

\end{exmp}

\begin{exmp} Next 
we execute \cref{alg:SDPandChoi} for positive maps $\Psi_{p,q}$ in \cref{thm:BS9}. We are minimizing 
\begin{multline} \label{eq:trace_pq}
\Tr \left(C(\Psi_{p,q})\, \rho \right) = \\
p^2 (1 - p q)^2 r_{00}  +  (2 p - q) q r_{11} +  
p^2 (2 p - q) q^3 r_{33} + q^2 (1 - p q)^2 r_{44} +  
(2 p - q) q r_{55}+ \\
 (2 p - q) q (1 - p^2 q^2) r_{66}  +  p^2 (2 p - q) q^3 r_{77} + 
  p^2 (1 - p q)^2 r_{88} -  \\
  p q (1 -q^2 + p^2 q^2) (r_{04} + \ols{r_{04}} +r_{48} + \ols{r_{48}})-
 (1 - p q) \left( 
  p^2 + (p- q)(1 - p^2)q  
 \right) 
    (r_{08} + \ols{r_{08}}). 
   \end{multline}
Since the objective function depends only on the real part of 
$\rho \in \M_9^{\text{sa}}$, we can use \cref{lem:sym}. 
 In \texttt{Mathematica} we run the semidefinite program for $\Psi_{p,q}$ and $\rho \in \M_9^{\text{sym}}$ 
   for 2100 points $(p,q) \in \mathcal{R}$. 
   The output states are all of the sparse form
   \begin{equation*}
    \rho_{p,q}=
   \left[   \begin{array}{ccc|ccc|ccc}
s_{00} & \cdot &  \cdot  &  
\cdot  & s_{04} & \cdot  & 
\cdot  &  \cdot   & s_{08}  \\
 \cdot  & s_{11} &  \cdot & \cdot &  \cdot  &  \cdot  &
   \cdot  &  \cdot  &  \cdot  \\
 \cdot  &  \cdot  &  s_{22} &
  \cdot  &  \cdot  &  \cdot & \cdot & \cdot  &  \cdot  \\
\hline
 \cdot  &\cdot &  \cdot  &
 s_{33} &  \cdot  &  \cdot  &
  \cdot  &  \cdot  &  \cdot  \\
s_{04} &   \cdot  &  \cdot  &
 \cdot & s_{44} & \cdot &
  \cdot  &  \cdot   & s_{48}  \\
 \cdot  &  \cdot &  \cdot  &
  \cdot  &  \cdot  & s_{55} &
   \cdot & \cdot & \cdot \\
\hline
 \cdot  &  \cdot  & \cdot &
  \cdot  &  \cdot  &  \cdot &
  s_{66}  &  \cdot  &  \cdot \\
 \cdot  &  \cdot  &  \cdot  &
  \cdot  &  \cdot  & \cdot & 
   \cdot & s_{77} &  \cdot \\
s_{08} & \cdot  &  \cdot  &
 \cdot  & s_{48} &   \cdot  &
  \cdot &  \cdot  & s_{88}
\end{array} \right].
   \end{equation*}
The output values 
$\Tr \left(C(\Psi_{p,q})\, \rho_{p,q} \right)$ are shown in \cref{fig:tracesPsi_pq}, where $(p,q) \in \mathcal{R}$ are the 2100 points in \cref{fig:leafpq}.
\begin{figure}[htbp]
  \centering
   \includegraphics[width=0.27\textwidth, height=50mm]{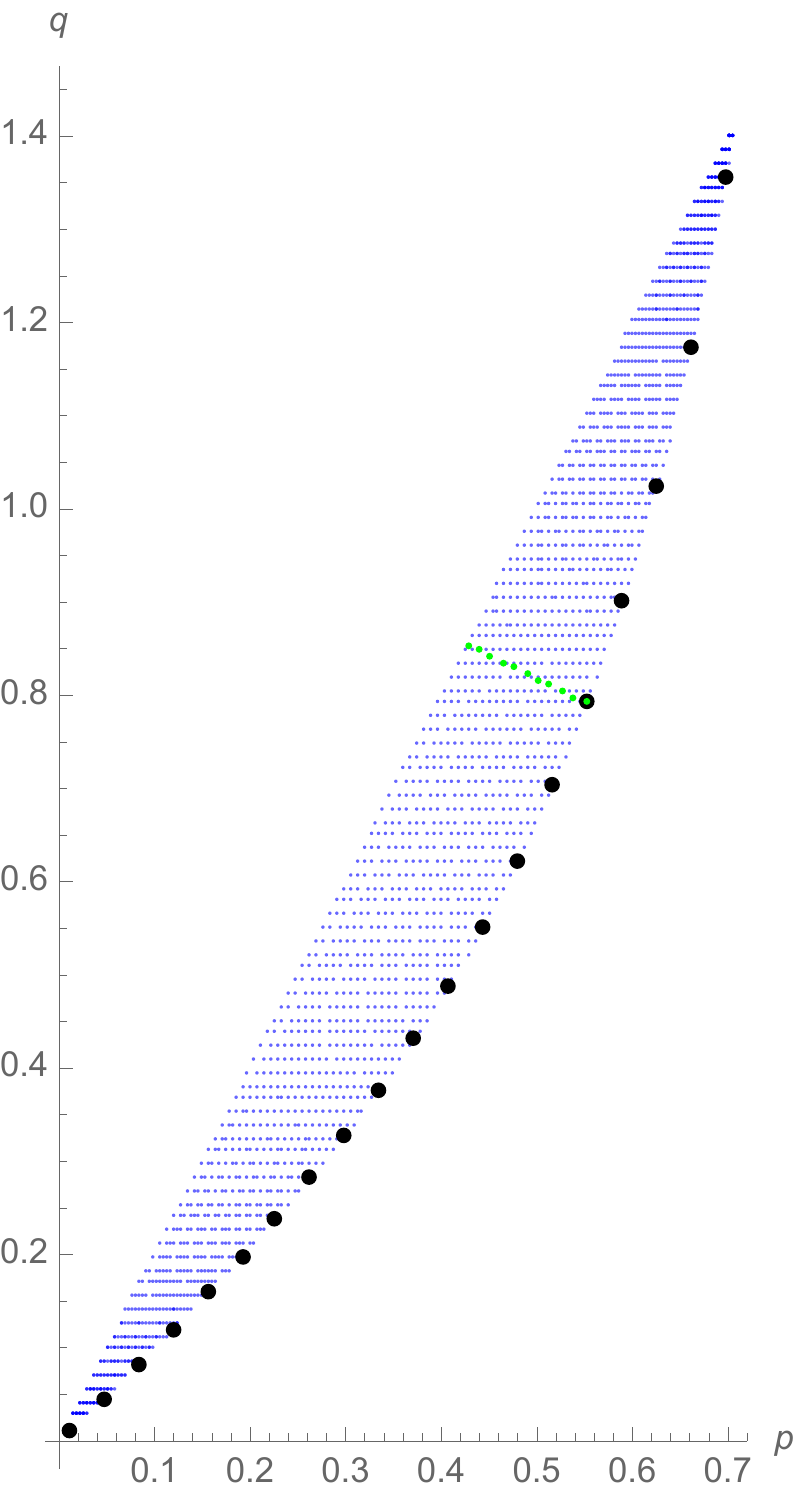}
  \includegraphics[width=0.70\textwidth, height=50mm]{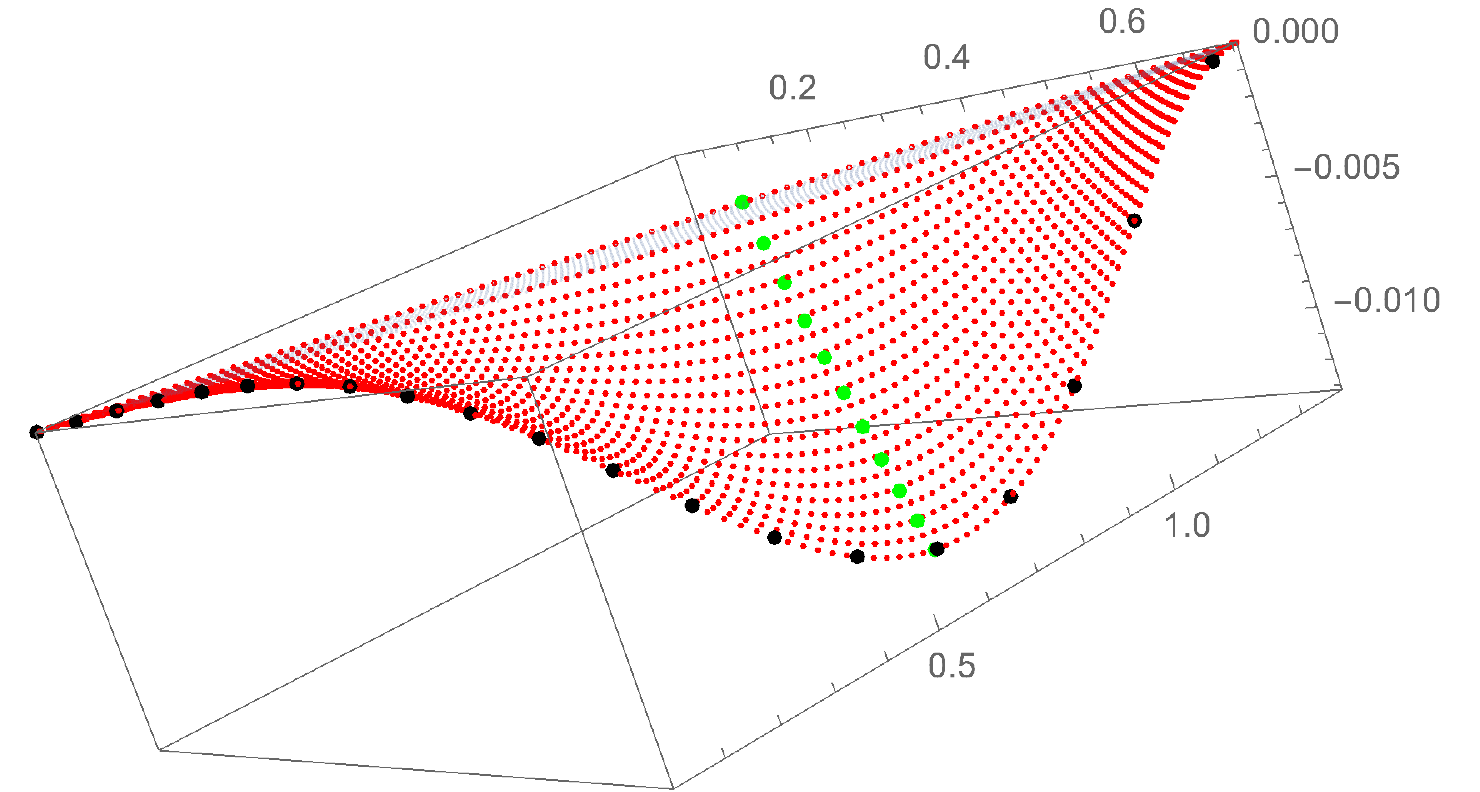}
  \put(-78,30){$\substack{p}$}
  \put(-110,127){$\substack{q}$}
  \put(-30,123){$\substack{\Tr}$}
  \caption{Plot of outputs of \cref{alg:SDPandChoi}:
  $\left( p, q, 
 \Tr \left[C(\Psi_{p,q})\, \rho_{p,q} \right] \right)$  for 2100 points
 $(p,q) \in \mathcal{R}$.}
  \label{fig:tracesPsi_pq}
\end{figure}
\begin{table}[htbp]
{\footnotesize
 \caption{All the values of $\Tr \left(C(\Psi_{p,q})\, \rho_{p,q} \right)$ are nonpositive}  
 \label{tab:rho_pq}
\begin{center}
  \begin{tabular}{|c|c|}\hline 
$(p,q)\in\mathcal{R}$ & $\min \Tr \left(C(\Psi_{p,q})\, \rho_{p,q} \right)$
\\ \hline 
\multicolumn{2}{|c|}{ $\rho_{p,q} \colon s_{00}, s_{11}, s_{22}, s_{33}, s_{44}, s_{55}, s_{66}, s_{77}, s_{88}, s_{04}, s_{08}, s_{48}$  }  \\ \hline \hline
  &  \\
 $\substack{(0.427039, 0.854079)}$ & 0   \\
\multicolumn{2}{|c|}{
$\substack{ 
0.10456, 0.13609, 0.11404, 0.12302, 0.04577, 0.12536, 0.11617, 0.13062, 0.10437, 0.04582, -0.01292, 0.04573}$  } \\ \hline
 $\substack{(0.439463, 0.84786)}$ & $\substack{-0.00160017}$   \\
\multicolumn{2}{|c|}{
$\substack{ 
0.18313,0.06871,0.16802,0.20581,0.09568,0.03886,0.01706,0.13273,0.08999,0.11892,0.05354,0.07182}$  } \\ \hline
$\substack{(0.451887, 0.841655)}$ & $\substack{-0.00313393}$   \\
\multicolumn{2}{|c|}{
$\substack{ 
0.18281,0.06947,0.16907,0.20530,0.09619,0.03923,0.01710,0.13123,0.08961,0.11942,0.05376,0.07175}$  } \\ \hline
$\substack{(0.464311, 0.835443)}$ & $\substack{-0.00460266}$   \\
\multicolumn{2}{|c|}{
$\substack{ 
0.18252,0.07023,0.17016,0.20474,0.09665,0.03958,0.01713,0.12974,0.08924,0.11991,0.05400,0.07166}$  } \\ \hline
$\substack{(0.476735, 0.829231)}$ & $\substack{-0.00600778}$   \\
\multicolumn{2}{|c|}{
$\substack{ 
0.18225,0.07099,0.17130,0.20415,0.09706,0.03992,0.01718,0.12826,0.08890,0.12038,0.05425,0.07156}$ } \\ \hline
$\substack{(0.501582, 0.816807)}$ & $\substack{-0.00863339}$   \\
\multicolumn{2}{|c|}{
$\substack{ 
0.18180,0.07249,0.17380,0.20289,0.09767,0.04052,0.01728,0.12527,0.08827,0.12128,0.05481,0.07125}$ } \\ \hline
$\substack{(0.514006, 0.810595)}$ & $\substack{-0.00985712}$   \\
\multicolumn{2}{|c|}{
$\substack{ 
0.18165,0.07324,0.17517,0.20223,0.09786,0.04078,0.01735,0.12373,0.08800,0.12170,0.05513,0.07103}$ } \\ \hline
$\substack{(0.52643, 0.804383)}$ & $\substack{-0.0110238}$   \\
\multicolumn{2}{|c|}{
$\substack{ 
0.18154,0.07399,0.17666,0.20156,0.09795,0.04099,0.01742,0.12214,0.08775,0.12212,0.05548,0.07076}$ } \\ \hline
$\substack{(0.538854, 0.798172)}$ & $\substack{-0.0121351}$   \\
\multicolumn{2}{|c|}{
$\substack{ 
0.18150,0.07475,0.17828,0.20088,0.09793,0.04116,0.01751,0.12047,0.08752,0.12254,0.05587,0.07042}$ } \\ \hline
$\substack{(0.551278, 0.79196)}$ & $\substack{-0.0131931}$   \\
\multicolumn{2}{|c|}{
$\substack{ 
0.18155,0.07551,0.18008,0.20020,0.09777,0.04126,0.01761,0.11869,0.08733,0.12295,0.05632,0.06998}$ } \\     \hline \hline
  &  \\    
$\substack{(0.082327, 0.0828888)}$ & $\substack{-0.0000147211 }$   \\
\multicolumn{2}{|c|}{
$\substack{ 
0.15636,0.06036,0.09497,0.30052,0.13139,0.02134,0.01037,0.15670,0.06797,0.13469,0.03139,0.05783
}$ } \\ \hline
    $\substack{(0.190818, 0.198028)}$ & $\substack{-0.000413618 }$   \\
\multicolumn{2}{|c|}{
$\substack{ 
0.15956,0.06105,0.09983,0.29387,0.12795,0.02255,0.01091,0.15433,0.06995,0.13394,0.03300,0.05900
}$ } \\ \hline
    $\substack{(0.299308, 0.32876 )}$ & $\substack{-0.00234474 }$   \\
\multicolumn{2}{|c|}{
$\substack{ 
0.16518,0.06265,0.10996,0.28002,0.12193,0.02506,0.01198,0.14958,0.07363,0.13245,0.03629,0.06123
}$ } \\ \hline
    $\substack{(0.407799, 0.489143)}$ & $\substack{-0.00696419}$   \\
\multicolumn{2}{|c|}{
$\substack{ 
0.17273,0.06592,0.12854,0.25600,0.11348,0.02956,0.01376,0.14100,0.07900,0.12991,0.04205,0.06456
}$ } \\ \hline
    $\substack{( 0.516289, 0.703923)}$ & $\substack{-0.0125464  }$   \\
\multicolumn{2}{|c|}{
$\substack{ 
0.18013,0.07235,0.16307,0.21696,0.10224,0.03753,0.01650,0.12578,0.08543,0.12529,0.05188,0.06870
}$ } \\ \hline
    $\substack{( 0.552453, 0.795129)}$ & $\substack{-0.0131932 }$   \\
\multicolumn{2}{|c|}{
$\substack{ 
0.18158,0.07563,0.18072,0.19960,0.09761,0.04139,0.01765,0.11843,0.08739,0.12286,0.05648,0.07002
}$ } \\ \hline
    $\substack{( 0.588616, 0.900671)}$ & $\substack{ -0.0122921 }$   \\
\multicolumn{2}{|c|}{
$\substack{ 
0.18177,0.07967,0.20306,0.17976,0.09232,0.04608,0.01888,0.10956,0.08891,0.11967,0.06191,0.07105
}$ } \\ \hline
    $\substack{(0.62478, 1.02482 )}$ & $\substack{  -0.00934857 }$   \\
\multicolumn{2}{|c|}{
$\substack{ 
0.17996,0.08461,0.23155,0.15744,0.08614,0.05171,0.02009,0.09894,0.08957,0.11542,0.06820,0.07152
}$ } \\ \hline
    $\substack{(0.660943, 1.17365 )}$ & $\substack{ -0.00459576 }$   \\
\multicolumn{2}{|c|}{
$\substack{ 
0.17515,0.09053,0.26800,0.13287,0.07879,0.05832,0.02112,0.08643,0.08879,0.10968,0.07523,0.07100
}$ } \\ \hline
    $\substack{(0.697107, 1.35613 )}$ & $\substack{-0.00033405 }$   \\
\multicolumn{2}{|c|}{
$\substack{ 
0.16618,0.09736,0.31449,0.10672,0.06994,0.06574,0.02165,0.07210,0.08581,0.10194,0.08252,0.06884
}$ } \\  \hline
  \end{tabular}
\end{center}
}
\end{table}
\Cref{tab:rho_pq} shows some example results of \cref{alg:SDPandChoi} for a selection of input points $(p,q) \in \mathcal{R}$, 
and the corresponding output density matrices $\rho_{p,q}$ together with the nonpositive values $\Tr \left(C(\Psi_{p,q})\, \rho_{p,q} \right)$.  
The top part of the table corresponds to the points on the segment $(p,1.0676-\frac{1}{2}p) \in \mathcal{R}$ (shown green in \cref{fig:tracesPsi_pq}).
The bottom part of the table corresponds to the boundary points 
$(p,\frac{p}{1-p^2})$
of $\mathcal{R}$  (shown black in \cref{fig:tracesPsi_pq}).
Recall that on the boundary segment $q=2p$ the maps $\Psi_{p,2p}$ are completely positive, therefore \cref{alg:SDPandChoi} returns 
 $ \Tr \left(C(\Psi_{p,2p})\, \rho_{p,2p} \right)=0$  and
no new entangled states.

We conclude the analysis by comparing the outputs of
\cref{alg:SDP,alg:SDPandChoi} for $\Psi_{p,q}$ and the unital map $\Psi_{p,q}^{\text{u}}$ defined in \eqref{eq:pqUnital}.
Since $\Psi_{p,q}(I)^{-1/2}$ is invertible, the outputs of \cref{alg:SDP} for $\Psi_{p,q}$ and $\Psi_{p,q}^{\text{u}}$ are related by \cref{lem:projA} as follows
$\rho \leftrightarrow  \left( \Psi(I)^{-1/2} \otimes I  \right) \, \rho \,
   \left( \Psi(I)^{-1/2} \otimes I  \right)$.
However, \cref{alg:SDPandChoi} yields some inherently different entangled states.
\end{exmp}

\begin{exmp} Finally 
we run \cref{alg:SDPandChoi} for positive maps $\Psi_{m,n}$ from \cref{thm:BS8}. On account of symmetries we only consider $(m,n)$ in the top quarter of the region $\mathcal{A}$ as it is shown in \cref{fig:leafmn}.
For $\rho=[r_{ij}]_{i,j=0}^8 \in \M_9^{\text{sa}}$ we are minimizing 
\begin{multline*} \label{eq:trace_mn}
\Tr \left(C(\Psi_{m,n})\, \rho \right) = \\
n^2 r_{00} -  (2 m n + n^2) r_{11} + m^2 n^2 r_{33} +  m^2 r_{44} + 
b( r_{55} +r_{66})+ 
  (b + (m + n)^2) r_{88}  \\
  + m n^2 (r_{01} + \ols{r_{01}}-r_{03} - \ols{r_{03}}) +
   m^2 n (r_{14} + \ols{r_{14}}- r_{34} - \ols{r_{34}}) 
    - m n (m + n) (r_{18} + \ols{r_{18}}-r_{38} -\ols{r_{38}})
   \\
+(m n +c) (r_{04} + \ols{r_{04}})
- (m^2 n^2 +c ) (r_{13} + \ols{r_{13}}) \\
  -  (b + n (m + n))  (r_{08} + \ols{r_{08}})
- (b + m (m + n)) (r_{48} + \ols{r_{48}}),
   \end{multline*}
which only depends on the real part of 
$\rho$, thus we can by \cref{lem:sym} execute the semidefinite program 
in \texttt{Mathematica}.
As inputs we take positive maps  $\Psi_{m,n}$
   for 4200 points $(m,n) \in \mathcal{A}$ as shown on \cref{fig:traces_mn}. 
   The output density matrices are real symmetric matrices of the form
   \begin{equation*}
    \rho_{m,n}=
   \left[   \begin{array}{ccc|ccc|ccc}
r_{00} & r_{01} &  \cdot  &  
r_{03}  & r_{04} & \cdot  & 
\cdot  &  \cdot   & r_{08}  \\
 r_{01}  & r_{11} &  \cdot & 
 r_{13}  & r_{14}  &  \cdot  &
   \cdot  &  \cdot  &  r_{18}  \\
 \cdot  &  \cdot  &  r_{22} &
  \cdot  &  \cdot  &   r_{25} &
  \cdot & \cdot  &  \cdot  \\
\hline
r_{03} & r_{13} &  \cdot  &  
r_{33}  & r_{34} & \cdot  & 
\cdot  &  \cdot   & r_{38}  \\
 r_{04}  & r_{14} &  \cdot & 
 r_{34}  & r_{44}  &  \cdot  &
   \cdot  &  \cdot  &  r_{48}  \\
 \cdot  &  \cdot &  r_{25} &
  \cdot  &  \cdot  & r_{55} &
   \cdot & \cdot & \cdot \\
\hline
 \cdot  &  \cdot  & \cdot &
  \cdot  &  \cdot  &  \cdot &
  r_{66}  &  r_{67}  &  \cdot \\
 \cdot  &  \cdot &  \cdot  &
  \cdot  &  \cdot  & \cdot & 
   r_{67} & r_{77} &  \cdot \\
r_{08} & r_{18}  &  \cdot  &
 r_{38}  & r_{48} &   \cdot  &
  \cdot &  \cdot  & r_{88}
\end{array} \right].
   \end{equation*}
The output values 
$\Tr \left(C(\Psi_{m,n})\, \rho_{m,n} \right)$, 
which are all nonpositive, are shown in \cref{fig:traces_mn}.
\begin{figure}[htbp]
  \centering
  \includegraphics[width=0.3\textwidth]{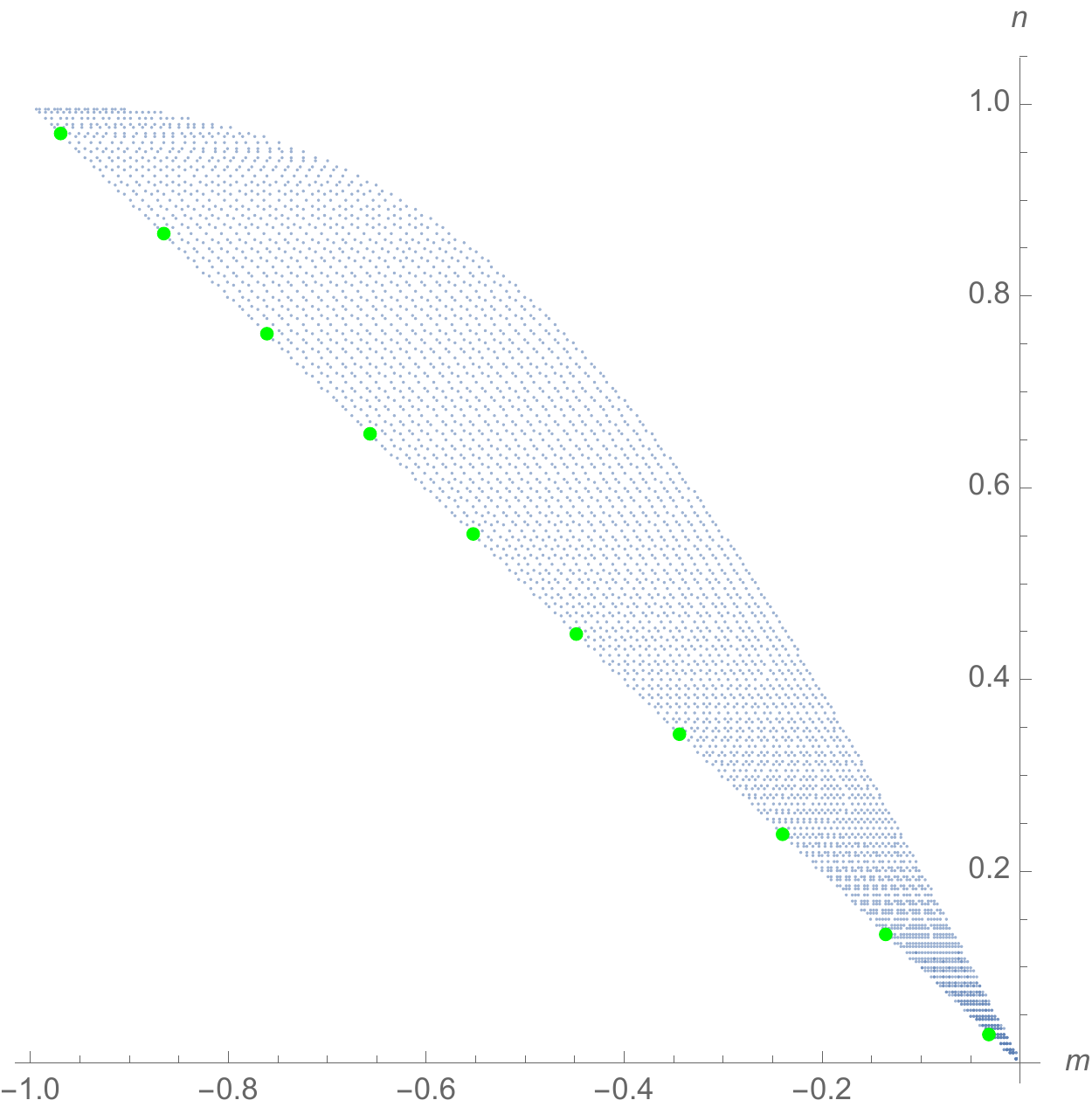}
   \includegraphics[width=0.6\textwidth]{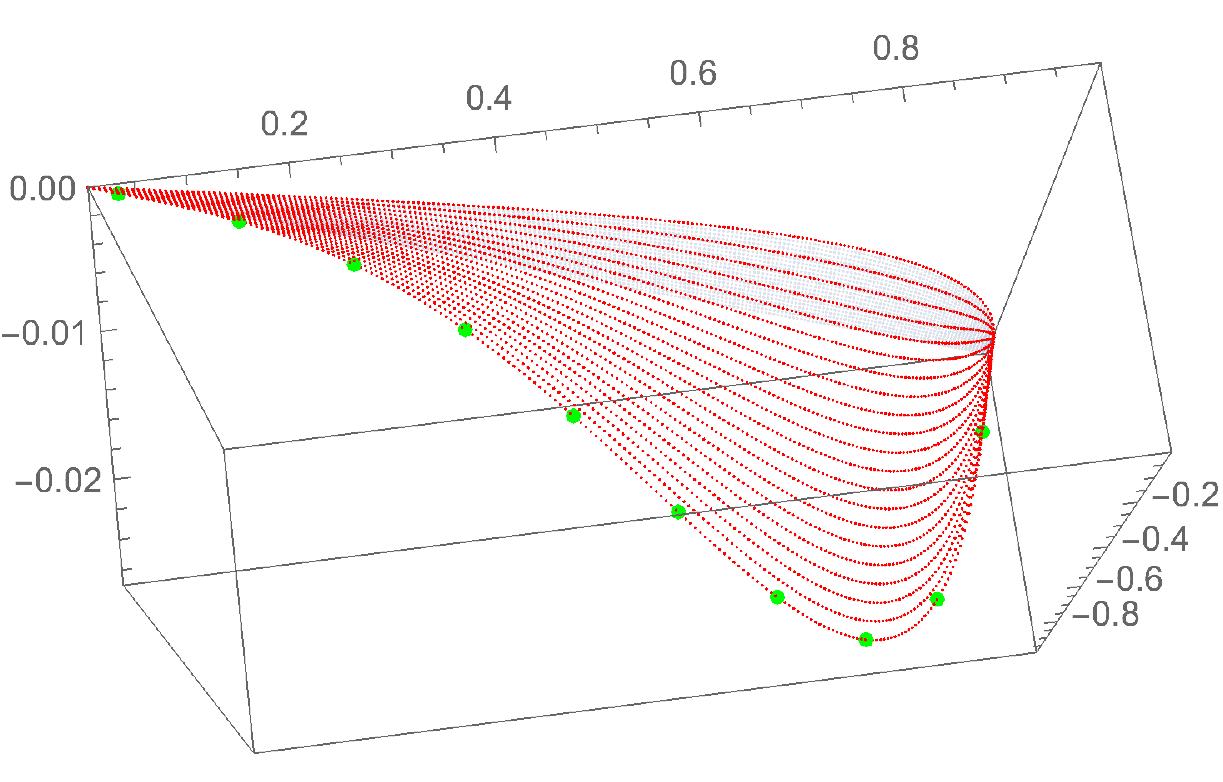}
   \put(-120,123){$\substack{n}$}
   \put(-32,40){$\substack{m}$}
   \put(-216,65){$\substack{\text{Tr}}$}
  \caption{Inputs $(m,n) \in  \mathcal{A}$  and outputs 
   $\left( m, n, 
 \Tr \left[C(\Psi_{m,n})\, \rho_{m,n} \right] \right)$ 
  of \cref{alg:SDPandChoi} 
  for $\Psi_{m,n}$. }
  \label{fig:traces_mn}
\end{figure}
\Cref{tab:rho_mn} contains some results of \cref{alg:SDPandChoi} 
for a selection of input points $(m,n) \in \mathcal{A}$ 
and the corresponding output density matrices $\rho_{m,n}$ 
(the corresponding points $\left( m, n, 
 \Tr \left[C(\Psi_{m,n})\, \rho_{m,n} \right] \right)$ are 
 shown green in \cref{fig:traces_mn}).
 
 \begin{table}[htbp]
{\footnotesize
 \caption{All the values of $\Tr \left(C(\Psi_{m,n})\, \rho_{m,n} \right)$ are nonpositive. }  
 \label{tab:rho_mn}
\begin{center}
  \begin{tabular}{|c|c|}\hline 
$(m,n)\in\mathcal{A}$ & $\min \Tr \left(C(\Psi_{m,n})\, \rho_{m,n} \right)$
\\ \hline 
\multicolumn{2}{|c|}{ $\rho_{m,n} \colon r_{00}, r_{11}, r_{22}, r_{33}, r_{44}, r_{55}, r_{66}, r_{77}, r_{88}, $  }  \\
\multicolumn{2}{|c|}{ \phantom{AAAAAAAAAA} $ r_{01}, r_{03}, r_{04}, r_{08}, r_{13}, r_{14}, r_{18}, r_{25}, r_{34}, r_{38}, r_{48}, r_{67}$  }  \\\hline \hline
 $\substack{(-0.97, 0.97)}$ & -0.00745894  \\
\multicolumn{2}{|c|}{
$\substack{ 
0.14921,0.03713,0.13783,0.25055,0.13227,0.04035,0.02639,0.16208,0.06417,  }$  } \\ 
\multicolumn{2}{|c|}{
$\substack{0.00065,-0.13956,0.01728,0.06009,0.00325,-0.02674,-0.00548,0.02036,0.09624,0.01157,0.07694,0.00323}$  } \\ \hline
  $\substack{ (-0.865556, 0.865556) }$ & -0.0230624  \\
\multicolumn{2}{|c|}{
$\substack{0.14923,0.03833,0.12543,0.26679,0.12891,0.04470,0.02677,0.15777,0.06208,}$  } \\ 
\multicolumn{2}{|c|}{
$\substack{-0.00376,-0.14500,0.02857,0.05772,-0.00195,-0.03110,-0.00685,0.02532,0.08814,0.01453,0.07756,0.00260}$  } \\ \hline
 $\substack{ (-0.761111, 0.761111) }$ & -0.0270046  \\
\multicolumn{2}{|c|}{
$\substack{0.14669,0.03788,0.11663,0.28154,0.12599,0.04829,0.02723,0.15511,0.06064,}$  } \\ 
\multicolumn{2}{|c|}{
$\substack{-0.01073,-0.14629,0.03949,0.05613,0.00070,-0.03274,-0.00809,0.02878,0.07886,0.01685,0.07779,0.00191}$  } \\ \hline
 $\substack{ (-0.656667, 0.656667) }$ & -0.0240555  \\
\multicolumn{2}{|c|}{
$\substack{0.14277,0.03754,0.10939,0.29419,0.12357,0.05160,0.02786,0.15348,0.05960,}$  } \\ 
\multicolumn{2}{|c|}{
$\substack{-0.01864,-0.14494,0.04993,0.05498,0.00710,-0.03396,-0.00933,0.03156,0.06860,0.01893,0.07785,0.00112}$  } \\ \hline
 $\substack{ (-0.447778, 0.447778) }$ & -0.0112646  \\
\multicolumn{2}{|c|}{
$\substack{0.13256,0.03944,0.09737,0.31112,0.12003,0.05787,0.02993,0.15305,0.05861,}$  } \\ 
\multicolumn{2}{|c|}{
$\substack{-0.03472,-0.13569,0.06872,0.05380,0.02520,-0.03792,-0.01210,0.03600,0.04615,0.02287,0.07771,-0.00119}$  } \\ \hline
 $\substack{ (-0.343333, 0.343333) }$ & -0.00574251  \\
\multicolumn{2}{|c|}{
$\substack{0.12693,0.04205,0.09206,0.31487,0.11865,0.06094,0.03156,0.15431,0.05862,}$  } \\ 
\multicolumn{2}{|c|}{
$\substack{-0.04220,-0.12826,0.07675,0.05376,0.03503,-0.04113,-0.01369,0.03787,0.03457,0.02486,0.07748,-0.00297}$  } \\ \hline
 $\substack{ (-0.238889, 0.238889) }$ & -0.00212082  \\
\multicolumn{2}{|c|}{
$\substack{0.12132,0.04574,0.08703,0.31536,0.11726,0.06397,0.03374,0.15666,0.05894,}$  } \\ 
\multicolumn{2}{|c|}{
$\substack{-0.04901,-0.11927,0.08370,0.05409,0.04443,-0.04514,-0.01542,0.03960,0.02317,0.02690,0.07709,-0.00539}$  } \\ \hline
 $\substack{ (-0.134444, 0.134444) }$ & -0.000404336  \\
\multicolumn{2}{|c|}{
$\substack{0.11593,0.05034,0.08220,0.31268,0.11561,0.06691,0.03657,0.16020,0.05957,}$  } \\ 
\multicolumn{2}{|c|}{
$\substack{-0.05499,-0.10897,0.08943,0.05479,0.05283,-0.04976,-0.01726,0.04121,0.01220,0.02902,0.07644,-0.00865}$  } \\ \hline
  \end{tabular}
\end{center}
}
\end{table}

\begin{remark}
For $(m,n)\in \mathcal{A}$ lying on the segment $m+n=0$, the objective function is much simpler,
\begin{multline*} 
\Tr \left(C(\Psi_{-n,n})\, \rho \right) = \\
n^2 \Big( r_{00} + r_{11} +  n^2 r_{33} + r_{44} + (1 \!-\! n^2) 
(r_{55} + r_{66} + r_{88}) 
-   2 n (r_{01} - r_{03} - r_{14} + r_{34})  \\
 -  (1 + n^2 + \sqrt{1 \!-\! n^2}) r_{04} - 
   2 (1 \!-\! n^2) (r_{08} + r_{48}) 
   - (1 + n^2 - \sqrt{1 \!-\! n^2}) r_{13}
  \Big).
   \end{multline*}
However, as seen in \cref{tab:rho_mn}, the optimal density matrices $\rho_{-n,n}$ do not simplify. Therefore we can not expect to find symbolic solutions to \cref{alg:SDPandChoi}. 
\end{remark}

For the points $(m,n)$  on the boundary of $\mathcal{A}$, \cref{alg:SDPandChoi} returns no new entangled states. Indeed, in \cref{pg:mnCP} we proved that for $(m,n) \in \partial \mathcal{A}$, the 
maps $\Psi_{m,n}$ completely positive with the  Kraus rank 1. The algorithm thus returns states for which  $ \Tr \left(C(\Psi_{m,n})\, \rho_{m,n} \right)=0$.

Finally, the comparison between the outputs of
\cref{alg:SDP,alg:SDPandChoi} for $\Psi_{m,n}$ and the corresponding unital entanglement witness $\Psi_{m,n}^{\text{u}}$ defined in \eqref{eq:mnUnital} is as in the previous example for $\Psi_{p,q}$.

\end{exmp}

\section{Conclusions}
\label{sec:conclusions}
We believe that our method of construction in \cref{sec:main}
can be generalised to higher dimensions.
Since we prescribe big infinite sets of zeros,
the linear system of equations yields families of positive maps that are easier to analyse. 

These examples may shed some light to Christandl's conjecture on entanglement breaking maps \cite{ChrConj} also for dimensions greater than 3.
 
Furthermore, our entanglement witnesses may contribute to understanding \textit{tensor stable positive maps} studied in \cite{Hermes} and \cite{Cuevas}.

\section*{Acknowledgments}
The author would like to acknowledge 
Alexander M\"uller-Hermes for suggesting to run an SDP on the positive maps in \cite{BucSiv} in order to obtain new entangled states. I thank 
the research groups of Gemma De las Cuevas and Tim Netzer, 
in particular to
Mirte van der Eyden and Andreas Klingler, whose questions helped to clarify the main ideas in the article.  
I am grateful to Seung-Hyeok Kye for pointing out the connection between
the parametrised examples of positive maps of the Choi type
and the positive maps in this article.

\bibliographystyle{siamplain}
\bibliography{article}

\end{document}